%% file: main.tex
\documentclass[fleqn]{lmcs}
\pdfoutput=1

\usepackage[utf8]{inputenc}

\usepackage{lastpage}
\lmcsdoi{17}{2}{15}
\lmcsheading{}{\pageref{LastPage}}{}{}%
{Feb.~19,~2018}{May~11,~2021}{}


\usepackage{amsmath} 

\usepackage{mystyle,mymath,proof,latexsym,amssymb}
\usepackage{color}
\usepackage{graphicx}

\usepackage{tikz}
\usetikzlibrary{positioning,shapes,calc}

\def\T#1{\hbox{\color{green}{$\clubsuit$ #1}}}

\begin{document}

\def\prooflineskip{\def\arraystretch{2.5}}

\input{macro_k}
\input{title}
\input{1}
\input{2}
\input{3}
\input{4}
\input{5}

\input{6}
\input{7}
\input{8}

\input{9}

\input{bib}

\end{document}

%% file: macro_k.tex
\def\SKIP{10pt}
\def\SSKIP{1pt} 

\def\iffDef{\overset{{\rm def}}{\Longleftrightarrow}}
\def\eqDef{\overset{{\rm def}}{=}}
\def\Underscore{\_}

\def\SLAR{\hbox{\bf SLA}}
\def\SLLA{\hbox{\bf SLAL}}
\def\SLG{\hbox{\bf G}}
\def\SLGL{\hbox{${\bf G}^+$}}
\def\Checker{\textbf{SLar}}
\def\ZZZ{\textbf{Z3}}

\def\Vars{\hbox{Vars}}
\def\Loc{\hbox{Loc}}
\def\Val{\hbox{Val}}
\def\FV{\hbox{FV}}
\def\Nil{\hbox{nil}}
\def\Emp{\hbox{Emp}}
\def\Pto{\mapsto}
\def\Arr{{\rm Arr}}
\def\Ls{{\rm ls}}
\def\Dll{{\rm dll}}
\def\Dom{{\rm Dom}}
\def\Ran{{\rm Ran}}
\def\Term{{\rm Tm}}
\def\Apply#1{{\rm Apply_{#1}}}
\def\Repeat#1{{\rm Repeat_{#1}}}
\def\Entl{\mathcal{E}}
\def\Pow{\mathcal{P}}
\def\R{\mathcal{R}}
\def\Store{{\rm Store}}
\def\Heap{{\rm Heap}}
\def\Sorted{{\rm Sorted}}
\def\mkPi{{\rm mkPi}}
\def\Perm{{\rm Perm}}
\def\pt{{\it pt}}
\def\T{\mathcal{T}}

\def\Tri#1#2#3{\{{#1}\}\hskip 0.5ex{#2}\hskip 0.5ex\{{#3}\}}

\def\Malloc{{\tt malloc}}
\def\Free{{\tt free}}
\def\If{{\tt if}\ }
\def\Else{\ {\tt else}\ }
\def\While{{\tt while}\ }

\def\False{\hbox{False}}
\def\True{\hbox{True}}

\def\Deriv#1#2#3{({\rm #1},#2,#3)}
\def\Search{{\tt search}}
\def\Fail{{\tt fail}}
\def\Valid{{\tt valid}}
\def\Invalid{{\tt invalid}}

\def\th{\widetilde{h}}
\def\wild{\hbox{-}}
\def\Cell#1#2{#1\to\downarrow#2}

\def\AssignK#1#2{#1\ :=\ #2}
\def\WhileK#1{{\bf while}(#1)}
\def\EndWhileK{{\bf end\ while}}
\def\FunctionK{{\bf function} }
\def\EndFunctionK{{\bf end\ function}}
\def\TakeK#1{{\bf take} #1}
\def\UnfinishedK{{\rm Unfinished}}
\def\FinishedK{{\rm Finished}}
\def\ReturnK#1{{\bf return}(#1)}
\def\IfK#1{{\bf if}\ (#1)}
\def\ElseK{{\bf else}\ }
\def\ElifK#1{{\bf else if}\ (#1)}
\def\EndIfK{{\bf end\ if}}

\def\qqOne{\quad}
\def\qqTwo{\quad\qqOne}
\def\qqThr{\quad\qqTwo}
\def\qqFor{\quad\qqThr}
\def\qqFiv{\quad\qqFor}
\def\qqSix{\quad\qqFiv}
\def\qqSev{\quad\qqSix}
\def\qqEig{\quad\qqSev}
\def\qqNin{\quad\qqEig}
\def\qqTen{\quad\qqNin}
\def\qqEle{\quad\qqTen}
\def\qqTwe{\quad\qqEle}
\def\qqTht{\quad\qqTwe}

\def\Deg{{\rm deg}}
\def\numPto{{\sharp_{\Pto}}}
\def\numList{{\sharp_{\rm lists}}}
\def\degUnfold#1#2{|#1|^{\rm Unfold}_{#2}}
\def\degEM#1#2{|#1|^{\rm EM}_{#2}}
\def\Rank#1#2{||#1||_{#2}}


\renewcommand{\theequation}{\alph{equation}}

\newcommand\mkcell[3]{{%
 \node[draw=black] (b1) at ($1.5*(#1,0)$) {\tiny #2};
 \node[draw=black, below = 0 of b1] {\tiny #3};
}}
\newcommand\mkarrow[2]{{%
 \draw[->] ($1.5*(#1,0)+(0.1,0)$) -- ++(1,0);
 \node [text width = 16pt] at ($1.5*(#1,0)+(0.7,0.2)$) {\tiny #2};
}}
\newcommand\mksarrow[2]{{%
 \draw[->] ($1.5*(#1,0)+(0.3,0)$) -- ++(0.8,0);
 \node [text width = 16pt] at ($1.5*(#1,0)+(0.7,0.2)$) {\tiny #2};
}}
\newcommand\mklarrow[2]{{%
 \draw[->] ($1.5*(#1,0)+(0.1,0)$) -- ++(2.5,0);
 \node [text width = 50pt] at ($1.5*(#1,0)+(1.3,0.2)$) {\tiny #2};
}}
\newcommand\mkarrowr[2]{{%
 \draw[<-] ($1.5*(#1,0)+(0.3,-0.45)$) -- ++(1.1,0);
 \node [text width = 16pt] at ($1.5*(#1,0)+(0.7,-0.65)$) {\tiny #2};
}}
\newcommand\mksarrowr[2]{{%
 \draw[<-] ($1.5*(#1,0)+(0.3,-0.45)$) -- ++(0.8,0);
 \node [text width = 16pt] at ($1.5*(#1,0)+(0.7,-0.65)$) {\tiny #2};
}}
\newcommand\mkdots[1]{{%
 \node at ($1.5*(#1,0)$) {\tiny ...};
}}
\newcommand\mkarrowout[2]{{%
 \draw[<-] ($1.5*(#1,0)+(0.3,-0.45)$) -- ++(0.8,0);
 \node [text width = 16pt] at ($1.5*(#1,0)+(0.7,-0.65)$) {\tiny #2};
}}

%% file: title.tex
\title[
Decidability for Entailments of 
Symbolic Heaps with Arrays
]{
Decidability for Entailments of 
Symbolic Heaps\texorpdfstring{\\}{ }with Arrays
}
\author[Daisuke Kimura]{Daisuke Kimura\rsuper{a}}
\address{\lsuper{a}Toho University, Chiba, Japan}
\email{kmr@is.sci.toho-u.ac.jp}

\author[Makoto Tatsuta]{Makoto Tatsuta\rsuper{b}}
\address{\lsuper{b}National Institute of Informatics, Tokyo, Japan}
\email{tatsuta@nii.ac.jp}

\maketitle

\begin{abstract}
This paper presents two decidability results on the validity checking problem
for entailments of symbolic heaps
in separation logic with Presburger arithmetic and arrays. 
The first result is for a system
with arrays and existential quantifiers. 
The correctness of the decision procedure is proved under the
condition that sizes of arrays in the succedent are not existentially quantified.
This condition is different from that proposed by Brotherston et al. in 2017
and one of them does not imply the other. 
The main idea is a novel translation from
an entailment of symbolic heaps into a formula in Presburger arithmetic. 
The second result is the decidability for a system with both arrays and lists. 
The key idea is to extend 
the unroll collapse technique proposed by Berdine et al. in 2005
to arrays and arithmetic as well as double-linked lists.
\end{abstract}

%% file: 1.tex
\section{Introduction}

Separation logic~\cite{Reynolds02} has been successfully used to verify/analyze heap-manipulating imperative programs with pointers
\cite{OHearn05,Brotherston17,OHearn06,OHearn11,Predator11},
and in particular it is successful for verify/analyze memory safety.
The advantage of separation logic is modularity brought by the frame rule,
with which we can independently verify/analyze each function that may
manipulate heaps~\cite{OHearn11}. 

In order to develop an automated analyzer/verifier of pointer programs based on separation logic,
symbolic-heap systems, which are fragments of separation logic, 
are often considered~\cite{OHearn04,OHearn05,OHearn11}. 
Despite of its simple and restricted form, symbolic heaps have enough expressive power, for example,
Infer~\cite{Infer11} and HIP/SLEEK
are based on symbolic-heap systems. 
Symbolic heaps
are used as assertions $A$ and $B$ in Hoare-triple 
$\{A\}P\{B\}$.
For program analysis/verification,
the validity of entailments need to be checked automatically. 

Inductive definitions in a symbolic-heap system is important,
since they can describe recursive data structures such as lists and trees.
Symbolic-heap systems with inductive predicates have been studied intensively~\cite{OHearn04,OHearn05,Kanovich14,Brotherston14,Cook11,Enea13,Enea14,Iosif13,Iosif14,Tatsuta15}. 
Berdine et al.~\cite{OHearn04,OHearn05} introduced a symbolic-heap system with hard-coded list and tree predicates, and showed the decidability of its entailment problem. Iosif et al.~\cite{Iosif13,Iosif14} considered a system with general inductive predicates, and showed its decidability under the bounded tree-width condition. Tatsuta et al.~\cite{Tatsuta15} introduced a system with general monadic inductive predicates.

Arrays are one of the most important primitive data structures of programs.
It is also important to verify that there is no buffer overflow in programs with arrays.
In order to verify/analyze pointer programs with arrays, 
this paper introduces two symbolic-heap systems with the array predicate.

The first symbolic-heap system, called $\SLAR$ (Separation Logic with Arrays), contains the points-to and the array predicates as well as existential quantifiers in spatial parts.
The entailments of this system also have disjunction in the succedents.
Our first main theorem is the decidability of the entailment problem.
For this theorem we need the condition:
the sizes of arrays in the succedent of an entailment do not contain any existential variables.
It means that the size of arrays in the succedent is completely determined by the antecedent. 

The basic idea of our decision procedure for $\SLAR$ is a novel translation of a given entailment into an equivalent formula in Presburger arithmetic. 
We use ``{\em sorted}'' symbolic heaps as a key idea for defining the translation. 
A heap represented by a sorted symbolic heap has addresses sorted in the order of the spatial part.
If both sides of a given entailment are sorted, 
the validity of the entailment is decided
by comparing spatial parts on both sides starting from left to right. 

The second system, called $\SLLA$ (Separation Logic with Arrays and Lists), contains the points-to predicate, the array predicate, the singly-linked list predicate, and the doubly-linked list predicate.
They also have disjunction in the succedents.
The entailments of this system are restricted to quantifier-free symbolic heaps.
This restriction comes from a technical reason, namely, our key idea (the unroll collapse technique) does not work in the presence of existential quantifiers.
Although the restriction reduces the expressive power of entailments, the quantifier-free entailments are useful for verify/analyze memory safety~\cite{OHearn04,OHearn05}. 

The second main theorem of this paper is the decidability of the entailment problem for $\SLLA$. 
Our decision procedure is split into the two stages (a) and (b): 
(a) the first stage eliminates the list predicates from 
the antecedent of a given entailment 
by applying {\em the unroll collapse technique}. 
Originally the unroll collapse for acyclic singly-linked list segments is 
invented by Berdine et al.~\cite{OHearn04}. 
We extend the original one in two ways;
the unroll collapse in $\SLLA$ is extended to arithmetic and arrays
as well as doubly-linked list segments.
(b) the second stage eliminates list predicates from 
the succedent of the entailment by the proof search technique. 
To do this, we introduce a sound and complete proof system that has valid quantifier-free entailments in $\SLAR$ as axioms,
which are checked by the first decision procedure for $\SLAR$.

As related work, as far as we know there are 
two papers about symbolic-heap systems that have arrays as primitive.
Calcagno et al.~\cite{OHearn06} studied program analysis based on a symbolic-heap system in the presence of pointer arithmetic.
They assumed a decision procedure for entailments with arrays and did not propose it.
Brotherston et al.~\cite{Brotherston17} considered the same system as $\SLAR$, and investigated several problems about it.

In \cite{Brotherston17}, they proposed a decision procedure for the entailment problem (of $\SLAR$) by
giving an equivalent condition to the existence of a counter-model for a given entailment, then 
checking a Presburger formula that expresses the condition. 
In order to do this, they imposed the restriction that the second argument of
the points-to predicate in the succedent of an entailment is not existentially quantified.
Their result decides a different class of entailments from the class decided by ours, and
our class neither contains their class nor is contained by it. 

When we extend separation logic with arrays, it may be different from previous array
logics in the points that (1) it is specialized for memory safety, and
(2) it can scale up by modularity. 
Bradley et al.~\cite{Bradley06}, Bouajjani et al.~\cite{Bouajjani09}, and Lahiri et al.~\cite{Lahiri08} 
discussed logics for arrays but their systems are totally
different from separation logic. So we cannot apply their techniques to our
case.
Piskac et al.~\cite{Piskac13} proposed a separation logic system with list
segments, and it can be combined with various SMT solvers, including
array logics. However, when we combine it with array logics, the arrays
are external and the resulting system does not describe the arrays by
spatial formulas with separating conjunction. 

The first result of this paper is based on \cite{Kimura17} but
we give detailed proofs of the key lemmas (stated in Lemma~\ref{lemma:correct}) for the result. 

The paper is organized as follows. 
Section 2 introduces the first system $\SLAR$.
Section 3 defines and discusses the decision procedure of the entailment problem for $\SLAR$.
In Section 4,
we show the first main theorem, namely, the decidability result of $\SLAR$.
In Section 5, we introduce the second system $\SLLA$. 
Section 6 shows the unroll collapse property.
Section 7 gives a decision procedure for $\SLLA$.
Section 8 shows the second main theorem, namely, the decidability result of $\SLLA$.
We conclude in Section 9.

%% file: 2.tex
\section{Separation Logic with Arrays}\label{sect:sla}

This section defines the syntax and semantics of our separation logic with arrays. 
We first give the separation logic $\SLG$ with arrays in the ordinary style. 
Then we define the symbolic-heap system $\SLAR$ as a fragment of $\SLG$. 

\subsection{Syntax of System $\SLG$ of Separation Logic with Arrays}

We use the following notations in this paper. 
Let $(p_j)_{j \in J}$ be a sequence indexed by a finite set $J$. 
We write $\{p_j\ |\ j \in J\}$ for this sequence. 
This sequence will sometimes be abbreviated by $\Vec{p}$. 
 We write $q \in \Vec{p}$ when $q$ is an element of $\Vec{p}$. 

We have first-order variables $x,y,z,\ldots\in \Vars$ and constants $0,1,2,\ldots$. 
The syntax of $\SLG$ is defined as follows: 
\begin{align*}
\hbox{Terms}
\quad 
&
t ::= x\ |\ 0\ |\ 1\ |\ 2\ | \cdots |\ t+t. 
\\
\hbox{Formulas}
\quad 
&
F ::= t = t\ |\ F \land F\ |\ \neg F\ |\ \exists xF\ |\ \Emp\ |\ t \mapsto (t,\ldots,t)\ |\ \Arr(t,t)\ |\ F * F.
\end{align*}

Atomic formulas $t\mapsto (u_1,\ldots,u_{\pt})$ and $\Arr(t,u)$ are called 
a points-to atomic formula and an array atomic formula, respectively.
The points-to predicate $\mapsto$ is $(\pt+1)$-ary predicate
(the number $\pt$ is fixed beforehand). 
We sometimes write $t \mapsto u$ instead of $t \mapsto (u)$ when $\mapsto$ is a binary predicate. 
We use the symbol $\wild$ to denote an unspecified term. 

Each formula is interpreted by a variable assignment and a heap: 
$\Emp$ is true when the heap is empty; 
$t \mapsto (\Vec{u})$ is true when
the heap has only a single memory cell of address $t$ that contains the value $\Vec{u}$; 
$\Arr(t,u)$ is true when the heap has only an array 
of index from $t$ to $u$; 
a separating conjunction $F_1 * F_2$ is true when
the heap is split into some two disjoint sub-heaps, 
$F_1$ is true for one, and $F_2$ is true for the other. 
The formal definition of these interpretations is given in the next subsection.

A term $t$ that appears in either $t \mapsto (\wild)$, $\Arr(t,\wild)$ or $\Arr(\wild,t)$ of $F$ 
is called an address term of $F$. 
The set of free variables (denoted by $\FV(F)$) of $F$ is defined as usual. 
We also define $\FV(\Vec{F})$ as the union of $\FV(F)$, where $F \in \Vec{F}$. 

We use abbreviations $F_1 \vee F_2$, $F_1 \to F_2$, and $\forall xF$ defined in the usual way.
We also write $t \neq u$, $t \le u$, $t < u$, $\True$, and $\False$ 
for $\neg (t = u)$, $\exists x(u = t + x)$, $t + 1 \le u$, 
$0 = 0$, and $0\neq 0$, respectively. 

A formula is said to be \textit{pure} if it 
is a formula of Presburger arithmetic.

We implicitly use the associative and commutative laws for the connectives $*$ and $\land$, 
the fact that $\Emp$ is the unit of $*$ and $\True$ is the unit of $\land$, 
and permutation of the existential quantifiers. 

For $I$ a finite set and $\{F_i\}_{i\in I}$ a set of formulas, 
we write $\bigwedge_{i \in I} F_i$ for the conjunction of the elements of $\{F_i\}_{i\in I}$.
If $I$ is empty, it is defined as $\True$. 

\subsection{Semantics of System $\SLG$ of Separation Logic with Arrays}

Let $N$ be the set of natural numbers. We define the following semantic domains: 
\[
\Val \eqDef N, 
\quad
\Loc \eqDef N\setminus\{0\}, 
\quad
\Store \eqDef \Vars \to \Val, 
\quad
\Heap \eqDef \Loc \to_{\rm fin} \Val^\pt.
\]

$\Loc$ means addresses of heaps. $0$ means Null. 
An element $s$ in $\Store$ is called a \textit{store}
that means a valuation of variables. 
The update $s[x_1:=a_1,\ldots,x_k:=a_k]$ of $s$ is defined by 
$s[x_1:=a_1,\ldots,x_k:=a_k](z) = a_i$ if $z = x_i$, 
otherwise $s[x_1:=a_1,\ldots,x_k:=a_k](z) = s(z)$. 
An element $h$ in $\Heap$ is called a \textit{heap}. 
The domain of $h$ (denoted by $\Dom(h)$) means the memory addresses which are currently used. 
$h(n)$ means the value at the address $n$ if it is defined.
We sometimes use notation $h_1+h_2$ for the disjoint union of $h_1$ and $h_2$,
that is, it is defined when $\Dom(h_1)$ and $\Dom(h_2)$ are disjoint sets, and 
$(h_1+h_2)(n)$ is $h_i(n)$ if $n \in \Dom(h_i)$ for $i = 1,2$. 
The restriction $h|_X$ of $h$ is defined by $\Dom(h|_X) = X \cap \Dom(h)$ 
and $h|_X(m) = h(m)$ for any $m \in \Dom(h|_X)$. 
A pair $(s,h)$ is called a \textit{heap model}. 

The interpretation $s(t)$ of a term $t$ by $s$ is defined by 
extending the definition of $s$ by $s(n) = n$ for each constant $n$, 
and $s(t+u) = s(t)+s(u)$.
We also use the notation $s(\Vec{u})$ defined by $s(u_1,\ldots,u_k) = (s(u_1),\ldots,s(u_k))$. 

The interpretation $s,h\models F$ of $F$ under the heap model $(s,h)$ 
is defined inductively as follows: 

$s,h\models t = u$
\ iff\ 
$s(t) = s(u)$, 

$s,h\models F_1\land F_2$
\ iff\ 
$s,h\models F_1$ and $s,h\models F_2$, 

$s,h\models \neg F$
\ iff\ 
$s,h\not\models F$, 

$s,h\models \exists xF$
\ iff\ 
$s[x:= a],h\models F$ for some $a \in \Val$, 

$s,h\models \Emp$
\ iff\ 
$\Dom(h) = \emptyset$, 

$s,h\models t \mapsto (\Vec{u})$
\ iff\ 
$\Dom(h) = \{s(t)\}$ and $h(s(t)) = s(\Vec{u})$, 

$s,h\models \Arr(t,u)$
\ iff\ 
$s(t) \le s(u)$ and $\Dom(h) = \{ x \in N\ |\ s(t) \le x \le s(u)\}$, 

$s,h\models F_1*F_2$
\ iff\ 
$s,h_1\models F_1$, $s,h_2\models F_2$, and $h = h_1+h_2$ for some $h_1$, $h_2$. 

We sometimes write $s\models F$ if $s,h\models F$ holds for any $h$. 
This notation is mainly used for pure formulas, 
since their interpretation does not depend on the heap-part of heap models. 
We also write $\models F$ if $s,h\models F$ holds for any $s$ and $h$.

The notation $F_1 \models F_2$ is an abbreviation of $\models F_1 \to F_2$, that is, $s,h\models F_1$ implies $s,h\models F_2$ for any $s$ and $h$. 

\subsection{Symbolic-Heap System with Arrays}

The symbolic-heap system $\SLAR$ is defined as a fragment of $\SLG$. 
The syntax of $\SLAR$ is given as follows. 
Terms of $\SLAR$ are the same as those of $\SLG$. 
Formulas of $\SLAR$ (called \textit{symbolic heaps}) have the following form:
\[
\phi ::= \exists \Vec x(\Pi \land \Sigma)
\]
\noindent where $\Pi$ is a pure formula of $\SLG$ 
and $\Sigma$ is the spatial part defined by 
\[
\Sigma ::= \Emp\ |\ t \mapsto (t,\ldots,t)\ |\ \Arr(t,t)\ |\ \Sigma * \Sigma.
\]

In this paper, we use the following notations for symbolic heaps. 
The symbol $\sigma$ is used to denote an atomic formula of $\Sigma$. 
We write $x \mapsto \Underscore$ for
$\exists z(x \mapsto z)$ where $z$ is fresh. 
We also write $\exists\Vec{x}(\Pi\land\Sigma)\land\Pi'$ for $\exists\Vec{x}(\Pi \land \Pi'\land\Sigma)$, 
write $\exists\Vec{x}(\Pi\land\Sigma)*\Sigma'$ for $\exists\Vec{x}(\Pi\land \Sigma*\Sigma')$, 
and 
write $\exists\Vec{x}(\Pi\land\Sigma) * \exists\Vec{x'}(\Pi'\land\Sigma')$ for 
$\exists\Vec{x}\Vec{x'}(\Pi\land\Pi'\land \Sigma*\Sigma')$.

In this paper, we consider \textit{entailments} of $\SLAR$ that have the form: 
\[
\phi \prove \{\phi_i\ |\ i \in I\}
\qquad 
\hbox{($I$ is a finite set)}. 
\]
The left-hand side of the symbol $\prove$ is called the antecedent.
The right-hand side of the symbol $\prove$ is called the succedent.
The right-hand side $\{\phi_i\ |\ i \in I\}$ of an entailment means the disjunction of the symbolic heaps $\phi_i$ ($i\in I$). 

An entailment $\phi \prove \{\phi_i\ |\ i \in I\}$ is said to be \textit{valid}
if
$\phi \models \bigvee\{\phi_i \ |\ i \in I\}$ holds. 

A formula of the form $\Pi\land\Sigma$ is called a \textit{QF symbolic heap} (denoted by $\varphi$).
Note that existential quantifiers may appear in the pure part of a QF symbolic heap. 
We can easily see that
$\exists\Vec{x}\varphi\models\Vec\phi$ is equivalent to
$\varphi\models \Vec\phi$.
So we often assume that the left-hand sides of entailments are QF symbolic heaps. 

We call entailments of the form $\varphi \prove \{\varphi_i\ |\ i \in I\}$ \textit{QF entailments}.

\subsection{Analysis/Verification of Memory Safety}

We intend to use our entailment checker as a part of our
analysis/verification system for memory safety.
We briefly explain it for motivating our entailment checker.

The target programming language is essentially the same as
that in \cite{Reynolds02} except we extend the allocation command $\Malloc$, 
which returns the first address of the allocated memory block or 
returns $\Nil$ if it fails to allocate. 
We define our programming language in programming language C style.
\begin{align*}
\hbox{Expressions}\quad 
&
e ::= x \ |\  0 \ |\  1 \ |\  2 \ldots \ |\ e+e.
\\
\hbox{Boolean expressions}\quad
&
b ::= e==e \ |\ e<e \ |\ b \&\& b \ |\ b \| b \ |\ !b.
\\
\hbox{Programs}\quad 
&
P ::= x=e; \ |\ \If (b) \{ P \} \Else \{ P \}; \ |\
\While (b) \{ P \}; \ |\ 
P\ P
\\
&
\qquad\qquad
\ |\ x=\Malloc(y); \ |\
x=*y; \ |\
*x=y; \ |\
\Free(x);.
\end{align*}

$x=\Malloc(y);$ allocates $y$ cells and set $x$ to
the pointer to the first cell. Note that this operation may fail
if there is not enough free memory.

Our assertion language is a disjunction of symbolic heaps, namely,
\[
\hbox{Assertions}\quad A ::= \phi_1 \lor \dots \lor \phi_n.
\]
We write $\phi * (\phi_1 \lor \ldots \lor \phi_n)$ for 
$(\phi * \phi_1 \lor \ldots \lor \phi * \phi_n)$, and 
write $\exists x(\phi_1 \lor \ldots \lor \phi_n)$ for 
$(\exists x\phi_1 \lor \ldots \lor \exists x\phi_n)$. 
The notation $\Pi \land (\phi_1 \lor \ldots \lor \phi_n)$ is defined similarly. 

In the same way as \cite{Reynolds02},
we use a triple $\Tri A P B$ that means that
if the assertion $A$ holds at the initial state and the program $P$
is executed, then (1) if $P$ terminates then the assertion $B$ holds
at the resulting state, and (2) $P$ does not cause any memory errors.

As inference rules for triples,
we have ordinary inference rules for Hoare triples including the consequence
rule, as well as the following rules (axioms) for memory operations.
We write $\Arr2(x,y)$ for $\exists z(\Arr(x,z) \land x+y=z+1)$.
$\Arr2(x,y)$ denotes the memory block at address $x$ of size $y$.
\[
\Tri A{x=\Malloc(y);}{
\exists x'(A[x:=x'] \land x=\Nil \lor
A[x:=x'] * \Arr2(x,y[x:=x']))},\\
\Tri{A * y \mapsto t}{x=*y;}{\exists x'(A[x:=x'] * y \mapsto t[x:=x'] \land
x=t[x:=x'])}, \\
\Tri{\exists x'(A * x \mapsto x')}{*x=y;}{A * x \mapsto y}, \\
\Tri{\exists x'(A * x \mapsto x')}{\Free(x);}{A},\ \hbox{where $x'$ is fresh}.
\]

In order to prove memory safety of a program $P$ under a precondition $A$,
it is sufficient to
show that
$\{ A \} P \{ \True \}$ is provable.

By separation logic with arrays, we can show a triple
$\Tri A{x=\Malloc(y);}{\exists x'(A[x:=x'] \land x=\Nil \lor
A[x:=x'] * \Arr2(x,y[x:=x']))}$,
but
it is impossible without arrays since $y$ in $\Malloc(y)$ is a variable.
With separation logic with arrays,
we can also show
$\Tri{\Arr(p,p+3)}{*p = 5;}{p \mapsto 5 * \Arr(p+1,p+3)}$.

For the consequence rule
\[
\infer{\Tri A P B}{\Tri {A'} P {B'}} \qquad \hbox{(if $A \imp A'$ and  $B' \imp B$)}
\]
we have to check the side conditions $A \imp A'$ and $B' \imp B$.
Let $A$ be $\phi_1 \lor \ldots \lor \phi_n$ and
$A'$ be $\phi'_1 \lor \ldots \lor \phi'_m$.
Then we will use our entailment checker to decide
$\phi_i \prove \phi'_1, \ldots, \phi'_m$ for all $1 \le i \le n$.

%% file: 3.tex
\section{Decision Procedure for $\SLAR$}

This section gives our decision procedure of the entailment problem for $\SLAR$ 
by introducing the key idea, namely sorted symbolic heaps, 
for the decision procedure, and defining the translation from sorted symbolic heaps 
into formulas in Presburger arithmetic. 
We finally state the decidability result for $\SLAR$, 
which is the first main theorem in this paper.

\subsection{Sorted Entailments}

This subsection describes {\em sorted} symbolic heaps. 
In this and the next sections, 
for simplicity, we assume that the number $\pt$ is $1$, that is, 
the points-to predicate is a binary one. 
The decidability result and the decision procedure in these sections 
can be straightforwardly extended to arbitrary $\pt$. 
In these two sections, 
we will not implicitly use 
the commutative law for $*$, or the unit law for $\Emp$
in order to define the following notations. 

We give a pure formula $t < \Sigma$, 
which means the first address expressed by $\Sigma$ is greater than $t$. 
It is inductively defined as follows: 
\begin{center}
\begin{tabular}{l@{\qquad\qquad}l}
$t < \Emp \eqDef \True$, 
&
$t < t_1 \mapsto (\wild) \eqDef t < t_1$, 
\\
$t < \Arr(t_1,\wild) \eqDef t < t_1$, 
&
$t < (\Emp * \Sigma_1) \eqDef t < \Sigma_1$,
\\
$t < (t_1 \mapsto (\wild) * \Sigma_1) \eqDef t < t_1$, 
&
$t < (\Arr(t_1,\wild) * \Sigma_1) \eqDef t < t_1$. 
\end{tabular}
\end{center}
Then we inductively define a pure formula $\Sorted(\Sigma)$, 
which means the address terms are sorted in $\Sigma$ as follows. 
\begingroup
\allowdisplaybreaks
\begin{align*}
\Sorted'&(\Emp) \eqDef \True, 
\\
\Sorted'&(t\mapsto u) \eqDef \True, 
\\
\Sorted'&(\Arr(t,u)) \eqDef t \le u,
\\
\Sorted'&(\Emp*\Sigma_1) \eqDef \Sorted'(\Sigma_1), 
\\
\Sorted'&(t\mapsto u * \Sigma_1) \eqDef t < \Sigma_1 \land \Sorted'(\Sigma_1), 
\\
\Sorted'&(\Arr(t,u) * \Sigma_1) \eqDef t \le u \land u < \Sigma_1 \land \Sorted'(\Sigma_1),
\\
\Sorted&(\Sigma) \eqDef 0 < \Sigma \land \Sorted'(\Sigma). 
\end{align*}
\endgroup
We define $\Sigma^\sim$ as $\Sorted(\Sigma) \land \Sigma$.
Let $\phi$ be $\exists\vec{x}(\Pi\land\Sigma)$. 
We write $\Tilde\phi$ or $\phi^\sim$ for 
$\exists\vec{x}(\Pi\land\Sigma^\sim)$. 
We call $\Tilde\phi$ a {\em sorted symbolic heap}. 

We define $\Perm(\Sigma)$ as the set of permutations of $\Sigma$ with respect to $*$.
A symbolic heap $\phi'$ is called a permutation of $\phi$ if 
$\phi=\exists\Vec{x}(\Pi\land\Sigma)$, 
$\phi'=\exists\Vec{x}(\Pi\land\Sigma')$ and 
$\Sigma'$ is a permutation of $\Sigma$. 
We write $\Perm(\phi)$ for the set of permutations of $\phi$. 
Note that 
$s,h\models\Tilde{\phi'}$ for some $\phi'\in\Perm(\phi)$
if and only if 
$s,h\models\phi$.
An entailment is said to be {\em sorted} if all symbolic heaps in its antecedent and succedent are sorted. 

The next lemma claims that checking the validity of entailments can be reduced 
to checking the validity of sorted entailments. 

\begin{lem}\label{lemma:split_sorted}
$s\models \Tilde{\varphi'}\to\Lor\{\Tilde{\phi'}\ |\ i\in I, \phi'\in\Perm(\phi_i)\}$
for all $\varphi' \in \Perm(\varphi)$
\\
if and only if
$s\models \varphi\to\bigvee_{i\in I}\phi_i$. 
\end{lem}
\begin{proof}
We first show the left-to-right part. 
Assume the left-hand side of the claim. 
Fix $\varphi'\in \Perm(\varphi)$ and suppose $s,h \models \Tilde{\varphi'}$.
Then we have $s,h \models \varphi$.
By the assumption, $s,h\models \phi_i$ for some $i\in I$. 
Hence we have $s,h\models\Lor\{\Tilde{\phi'}\ |\ i\in I, \phi'\in\Perm(\phi_i)\}$. 
Next we show the right-to-left part.
Assume the right-hand side and $s,h\models\varphi$.
We have $s,h\models\Tilde{\varphi'}$ for some $\varphi'\in\Perm(\varphi)$. 
By the assumption, $s,h\models\Tilde{\phi'}$ for some $\phi'\in\Perm(\phi_i)$. 
Thus we have $s,h\models\phi_i$ for some $i\in I$. 
\end{proof}

The basic idea of our decision procedure is as follows:
(1) A given entailment is decomposed into sorted entailments according to Lemma~\ref{lemma:split_sorted}; 
(2) the decomposed sorted entailments are translated into Presburger formulas by the translation $P$ given in the next subsection; 
(3) the translated formulas are decided by the decision procedure of Presburger arithmetic. 

\subsection{Translation $P$}
We define the translation $P$ from QF entailments into Presburger formulas. 
We note that the resulting formula may contain new fresh variables (denoted by $z$).
In the definition of $P$, we fix a linear order on an index set $I$
to take an element of the minimum index. 
For saving space, we use some auxiliary notations.
Let $\{t_j\}_{j\in J}$ be a set of terms indexed by a finite set $J$. 
We write $u = t_J$ for $\bigwedge_{j\in J} u = t_j$. 
We also write $u < t_J$ for $\bigwedge_{j\in J} u < t_j$.
We note that both $u = t_\emptyset$ and $u < t_\emptyset$ are $\True$,
since $\bigwedge_{j\in \emptyset} u = t_j$ and $\bigwedge_{j\in \emptyset} u < t_j$ are defined by $\True$. 

The definition of $P(\Pi,\Sigma,S)$ is given as listed in Fig.~\ref{fig:transP}, where $S$ is a finite set $\{(\Pi_i,\Sigma_i)\}_{i \in I}$. 
We assume that pattern-matching is done from the top to the bottom. 

In order to describe the procedure $P$,
we temporarily extend terms to include  $u - t$ 
where $u,t$ are terms.
In the result of $P$, which is a 
Presburger arithmetic formula,
we eliminate these extended terms by replacing
$t'+(u-t)=t''$ and $t'+(u-t)<t''$
by
$t'+u = t''+ t$ and $t'+u < t'' + t$, respectively.

Let $\sharp_*(\Sigma)$ be the number of $*$ in $\Sigma$. 
Let $\sharp_*(\{(\Pi_i,\Sigma_i)\}_{i\in I})$ be
$\sum_{i \in I}\sharp_*(\Sigma_i)$. 
We define ${\rm FirstRemove}_\Sigma(\Sigma')$ by
$\Sigma'_0$ if $\Sigma$ has the form $t\Pto u * \Sigma_0$ and
$\Sigma'$ has the form $t'\Pto u' * \Sigma'_0$, or 
$\Sigma'$ otherwise. 
This is obtained as follows: 
the points-to predicates at the first positions of $\Sigma$ and $\Sigma'$
are removed (if they exist), then the resulting formula made
from $\Sigma'$ is returned.
${\rm FirstRemove}_\Sigma(\{(\Pi_i,\Sigma_i)\}_{i\in I})$ is also defined by
$\{(\Pi_i,{\rm FirstRemove}_\Sigma(\Sigma_i))\}_{i\in I}$.

The $\eqDef$ steps terminate 
since
the measure $(\sharp_*({\rm FirstRemove}_\Sigma(S)), \sharp_*(\Sigma)+\sharp_*(S), |S|)$
for $P(\Pi,\Sigma,S)$ strictly decreases, 
where $|S|$ is the number of elements in $S$.

\begin{figure}[t]\small
  \rule{\textwidth}{1pt}
  \\[10pt]
$
\begin{array}{lllr}
  P(\Pi,\Emp * \Sigma,S)
  &\eqDef&
  P(\Pi,\Sigma,S)
&
\hspace{75pt}
  {\bf (EmpL)}
  \\
  P(\Pi,\Sigma,\{( \Pi',\Emp * \Sigma')\} \cup S)
  &\eqDef&
  P(\Pi,\Sigma,\{( \Pi',\Sigma')\} \cup S)
&
  {\bf (EmpR)}
  \\
  P(\Pi,\Emp,\{( \Pi',\Sigma')\} \cup S)
  &\eqDef&
  P(\Pi,\Emp,S),
  \quad \hbox{where $\Sigma' \not\equiv \Emp$}
&
  {\bf (EmpNEmp)}
  \\
  P(\Pi,\Emp,\{( \Pi_i,\Emp)\}_{i \in I})
  &\eqDef&
  \Pi \imp \bigvee_{i \in I} \Pi_i
&
  {\bf (EmpEmp)}  
  \\
  P(\Pi,\Sigma,\{( \Pi', \Emp)\} \cup S)
  &\eqDef&
  P(\Pi,\Sigma,S), 
  \qquad\hbox{where $\Sigma \not\equiv \Emp$}
&
  {\bf (NEmpEmp)}  
  \\
  P(\Pi,\Sigma,\emptyset)
  &\eqDef&
  \neg(\Pi \land \Sorted(\Sigma))
&
  {\bf (empty)}    
\end{array}
$
\\[5pt]
$P(\Pi,t \mapsto u * \Sigma,\{( \Pi_i, t_i \mapsto u_i * \Sigma_i)\}_{i \in I})$
\hfill{($\mapsto\mapsto$)}
\\
\begin{tabular}[t]{ll}
  $\eqDef$&
  $P(\Pi \land t < \Sigma,\Sigma, \{( \Pi_i \land t=t_i \land u=u_i \land t_i < \Sigma_i,\Sigma_i)\}_{i \in I})$
\end{tabular}  
\\[5pt]
$P(\Pi,t \mapsto u * \Sigma,\{( \Pi_i,\Arr(t_i,t_i') * \Sigma_i)\} \cup S)$
\hfill{\bf ($\mapsto$Arr)}
\\
\begin{tabular}[t]{ll}
  $\eqDef$&
  $P(\Pi \land t_i'=t_i,t \mapsto u * \Sigma,\{(\Pi_i,t_i \mapsto u * \Sigma_i)\} \cup S)$
  \\
  &
  $\land\
  P(\Pi \land t_i'>t_i,t \mapsto u * \Sigma,\{(\Pi_i,t_i \mapsto u * \Arr(t_i+1,t_i') * \Sigma_i)\} \cup S)$
  \\
  &
  $\land\
  P(\Pi \land t_i'<t_i,t \mapsto u * \Sigma,S)$
\end{tabular}
\\[5pt]
$P(\Pi,\Arr(t,t') * \Sigma,S)$
\hfill{\bf (Arr$\mapsto$)}
\\
\begin{tabular}[t]{ll}
  $\eqDef$&
  $P(\Pi \land t' > t,t \mapsto z * \Arr(t+1,t') * \Sigma,S)$
  \\
  &
  $\land\ P(\Pi \land t' = t,t \mapsto z' * \Sigma,S)$,
  \quad
  \hbox{where $(\Pi'', t'' \mapsto u'' * \Sigma'') \in S$ and $z,z'$ are fresh}
\end{tabular}
\\[5pt]
$P(\Pi,\Arr(t,t') * \Sigma,\{( \Pi_i,\Arr(t_i,t_i') * \Sigma_i)\}_{i \in I})$
\hfill{\bf (ArrArr)}
\\
\begin{tabular}[t]{ll}
  $\eqDef$&
  $\Land_{I' \subseteq I}P\left(
  \begin{array}{l}
    \Pi \land
    m = m_{I'} \land m < m_{I\setminus I'}
    \land t \le t' \land t' < \Sigma \land 0 < t,
    \Sigma,
    \\
    \{( \Pi_i \land t = t_i \land t'_i < \Sigma_i, \Sigma_i)\}_{i \in I'}
    \cup
    \{( \Pi_i \land t = t_i,\Arr(t_i+m+1,t_i') * \Sigma_i)\}_{i \in I\setminus I'}
  \end{array}
  \right)$
  \\
  \multicolumn{2}{l}{
  $\land
  \Land_{\emptyset \ne I' \subseteq I}
  P\left(
  \begin{array}{l}
    \Pi
    \land m'<m \land m' = m_{I'} \land  m'<m_{I\setminus I'} \land 0 < t,
    \Arr(t+m'+1,t') * \Sigma,
    \\
    \{( \Pi_i \land t = t_i \land t'_i < \Sigma_i, \Sigma_i )\}_{i \in I'}
    \cup
    \{( \Pi_i \land t = t_i,\Arr(t_i+m'+1,t_i') * \Sigma_i )\}_{i \in I\setminus I'}
  \end{array}
  \right)$, 
  }
\end{tabular}
\\
where $m$, $m_i$, and $m'$ are abbreviations of $t'-t$, $t_i'-t_i$, and $m_{\min I'}$, respectively. 
\rule{\textwidth}{1pt}
\caption{The translation $P$}
\label{fig:transP}
\end{figure}

The formula $P(\Pi,\Sigma,\{(\Pi_i,\Sigma_i)\}_{i\in I})$ means 
that the QF entailment 
$\Pi\land\Tilde\Sigma \vdash \{\Pi_i\land\Tilde\Sigma\}_{i \in I}$ is valid.
From this intuition, we sometimes call $\Sigma$ the left spatial formula, and also call $\{\Sigma_i\}_{i\in I}$ the right spatial formulas.
We call the left-most position of a spatial formula
the head position.
The atomic formula appears at the head position is called the head atom. 

We will explain the meaning of each clause in Figure~\ref{fig:transP}. 

The clauses {\bf (EmpL)} and {\bf (EmpR)} just remove $\Emp$
at the head position. 

The clause {\bf (EmpNEmp)} handles the case where
the left spatial formula is $\Emp$.
A pair $(\Pi',\Sigma')$ in the third argument of $P$ is removed
if $\Sigma'$ is not $\Emp$,
since $\Pi'\land\Sigma'$ cannot be satisfied by the empty heap. 

The clause {\bf (EmpEmp)} handles the case where 
the left formula and all the right spatial formulas are $\Emp$.
This case $P$ returns a Presburger formula
which is equivalent to the corresponding entailment is valid. 

The clause {\bf (NEmpEmp)} handles the case where
the left spatial formula is not $\Emp$ and
a pair $(\Pi',\Emp)$ appears in the third argument of $P$. 
We remove the pair since $\Pi'\land\Emp$ cannot be satisfied
by any non-empty heap. 
For example,
$P(\True,x\mapsto 0 * y\mapsto 0,\{(\True,\Emp)\})$
becomes 
$P(\True,x\mapsto 0 * y\mapsto 0,\emptyset)$. 

The clause {\bf (empty)} handles the case where
the third argument of $P$ is empty. 
This case $P$ returns a Presburger formula which is equivalent to 
that the left symbolic heap $\Pi\land\Sigma$ is not satisfiable. 
For example,
$P(\True,x\mapsto 0 * y\mapsto 0,\emptyset)$
returns 
$\neg(x < y)$ using the fact that $\True$ is the unit of $\land$. 

The clause {\bf ($\mapsto\mapsto$)} handles the case where
all the head atoms of $\Sigma$ and $\{\Sigma_i\}_{i\in I}$ are
the points-to atomic formulas. 
This case we remove all of them and put equalities on the right pure parts. 
By this rule the measure is strictly reduced. 
For example,
$P(\True,3\mapsto y * 4\mapsto 11,\{(\True,x\mapsto y' * \Arr(4,4))\})$
becomes
$P(3<4,4\mapsto 11,\{(3=x \land y=y'\land x<4, \Arr(4,4))\})$. 

The clause {\bf ($\mapsto$Arr)} handles the case where
the head atom of the left spatial formula is the points-to atomic formula and
some right spatial formula $\Sigma_i$ has the array atomic formula as its head atom.
Then we split the array atomic formula into a points-to atomic formula and an array atomic formula for the rest.
We have three subcases according to the size of the head array. 
The first case is when the size of the array is $1$: 
We replace the head array by a points-to atomic formula. 
The second case is when the size of the head array is
greater than $1$: 
We split the head array atomic formula into 
a points-to atomic formula and an array atomic formula for the rest. 
The last case is when the size of the head array is less than $1$: 
We just remove $(\Pi_i,\Sigma_i)$, since the array atomic formula is false. 
We note that this rule can be applied repeatedly until all head array atomic formulas of
the right spatial formulas are unfolded,
since the left spatial formula is unchanged.
Then the measure is eventually reduced by applying {\bf ($\mapsto\mapsto$)}. 
For example,
$P(\True,x\mapsto 10,\{(\True,\Arr(4,5))\})$
becomes
\begin{align*}
P&(5=4, x\mapsto 10,\{(\True,4\mapsto 10)\})
\\
&\land
P(5>4, x\mapsto 10,\{(\True,4\mapsto 10 * \Arr(5,5))\})
\\
&\land
P(5<4, x\mapsto 10,\emptyset).
\end{align*}

The clause {\bf (Arr$\mapsto$)} handles the case where
the head atom of the left spatial formula is an array atomic formula and 
there is a right spatial formula whose head atom is a points-to atomic formula. 
We have two subcases according to the size of the head array.
The first case is the case where the size of the array is greater than $1$: 
The array atomic formula is split into 
a points-to atomic formula (with a fresh variable $z$)
and an array atomic formula for the rest. 
The second case is when the size of the array is $1$: 
The array atomic formula is unfolded and replaced by a points-to atomic formula with
a fresh variable $z'$. 
We note that the left head atom becomes a points-to atomic formula
after applying this rule. 
Hence the measure is eventually reduced,
since {\bf ($\mapsto\mapsto$)} or {\bf ($\mapsto$Arr)} will be applied next. 
For example,
$P(\True,\Arr(x,3),\{(\True,y\mapsto 10)\})$
becomes
\begin{align*}
P&(x<3, x\mapsto z*\Arr(x+1,3),\{(\True,y\mapsto 10)\})
\\
&\land
P(x=3, x\mapsto z',\{(\True, y\mapsto 10)\}). 
\end{align*}

The last clause {\bf (ArrArr)} handles the case where
all the head atoms in the left and right spatial formulas are array atomic formulas. 
We first find the head arrays of the shortest length among the head arrays. 
Next we split each longer array into two arrays
so that the first part has the same size as the shortest array. 
Then we remove the first parts. The shortest arrays are also removed. 
In this operation we have two subcases: 
The first case is when the array atomic formula of the left spatial formula
has the shortest size and disappears by the operation. 
The second case is when the array atomic formula of the left spatial formula
has a longer size, it is split into two array atomic formulas, and the second part remains.
We note that the measure is strictly reduced,
since at least one shortest array atomic formula is removed.
For example,
$P(\True,\Arr(x,5),\{(\True,\Arr(y,2)*\Arr(3,z))\})$
becomes
\begin{align*}
P&(5-x=2-y \land x\le 5 \land 0 < x, \Emp,\{(x=y \land 2<3,\Arr(3,z))\})
\\
&\land
P(5-x<2-y \land x\le 5 \land 0 < x, \Emp,\\
&\qquad\{(x=y,\Arr(y+(5-x)+1,2) * \Arr(3,z))\})
\\
&\land
P(2-y<5-x \land 0 < x, \Arr(x+(2-y)+1,5),\{(x=y \land 2<3,\Arr(3,z))\}).
\end{align*}
This example is a case when $I$ is a singleton set (we write $\{0\}$ for it),
$\Pi$ and $\Pi_0$ are $\True$, $\Sigma$ is $\Emp$, $\Sigma_0$ is $\Arr(3,z)$, 
$t$ is $x$, $t'$ is $5$, $t_0$ is $y$, $t'_0$ is $2$ in {\bf (ArrArr)}. 
The first clause of this result is obtained from the case $I'=I$ of the first conjunct in {\bf (ArrArr)}, 
namely, $m = m_{I'}$ is $5-x = 2-y$ and $m < m_{I\setminus I'}$ is $\True$. 
The second clause is obtained from the case $I'=\emptyset$ of the first conjunct, 
namely, $m = m_{I'}$ is $\True$ and $m < m_{I\setminus I'}$ is $5-x < 2-y$.
The last clause is obtained from the case $I'=I$ of the second conjunct in {\bf (ArrArr)},
namely, $m' < m$ is $2-y < 5-x$, $m' = m_{I'}$ is $2-y = 5-x$,
and $m'<m_{I \setminus I'}$ is $\True$.

\begin{exa}
The sorted entailment
$(x\mapsto 10 * v\mapsto 11)^\sim \vdash \Arr(x',v')^\sim$
is translated by computing
$P(\True,x\mapsto 10 * v\mapsto 11, \{(\True,\Arr(x',v'))\})$. 
We will see its calculation step by step. 
It first becomes
\begin{align*}
  P&(x'=v',x\mapsto 10 * v\mapsto 11,\{(\True,x'\mapsto 10)\})
  \tag{a}
  \\
  &\land
  P(x'<v',x\mapsto 10 * v\mapsto 11,\{(\True,x'\mapsto 10 * \Arr(x'+1,v'))\})
  \tag{b}
  \\
  &\land
  P(x'>v',x\mapsto 10 * v\mapsto 11,\emptyset)
  \tag{c}
\end{align*}
by ($\mapsto$Arr)
taking $\Pi$ to be $\True$, $t\mapsto u$ to be $x\mapsto 10$, $\Sigma$ to be $v\mapsto 11$,
$\Pi_i$ to be $\True$, $\Arr(t_i,t'_i)$ to be $\Arr(x',v')$, $\Sigma_i$ to be $\Emp$, and $S$ to be empty. 
The first conjunct (a) becomes 
$P(x'=v'\land x<v, v\mapsto 11,\{(x'=x\land 10=10,\Emp)\})$
by $(\mapsto\mapsto)$
taking $\Pi$ to be $x'=v'$, $t\mapsto u$ to be $x\mapsto 10$, $\Sigma$ to be $v\mapsto 11$,
$I$ to be the singleton set $\{0\}$, 
$\Pi_0$ to be $\True$, $t_0\mapsto u_0$ to be $x'\mapsto 10$, and $\Sigma_0$ to be $\Emp$.
Then it becomes
\begin{align*}
  \neg(x'=v'\land x<v)
  \tag{a'}
\end{align*}
by (NEmpEmp) and (empty).
The third conjunct (c) becomes
\begin{align*}
  \neg(x'>v'\land x<v)
  \tag{c'}
\end{align*}
  by (empty)
taking $\Pi$ to be $x'>v'$ and $\Sigma$ to be $x\mapsto 10 * v\mapsto 11$.
The second conjunct (b) becomes 
$P(x'<v' \land x<v, v\mapsto 11,\{(x=x'\land 10=10,\Arr(x'+1,v'))\})$ by $(\mapsto\mapsto)$
taking
$\Pi$ to be $x'<v'$,
$t\mapsto u$ to be $x\mapsto 10$,
$\Sigma$ to be $v\mapsto 11$,
$I$ to be the singleton set $\{0\}$,
$\Pi_0$ to be $\True$,
$t_0\mapsto u_0$ to be $x'\mapsto 10$,
and $\Sigma_0$ to be $\Arr(x'+1,v')$,
then it becomes
\begin{align*}
  P&(x'<v' \land x<v\land x'+1=v', v\mapsto 11,\{(x=x'\land 10=10,x'+1\mapsto 11)\})
  \tag{b1}
  \\
  &\land
  P(x'<v' \land x<v\land x'+1<v', v\mapsto 11,\\
  &\qquad \{(x=x'\land 10=10,x'+1\mapsto 11 * \Arr(x'+2,v'))\})
  \tag{b2}
  \\
  &\land
  P(x'<v' \land x<v\land x'+1>v', v\mapsto 11,\emptyset)
  \tag{b3}
\end{align*}
by ($\mapsto$Arr)
taking
$\Pi$ to be $x'<v'\land x<v$,
$t\mapsto u$ to be $v\mapsto 11$,
$\Sigma$ to be $\Emp$,
$\Pi_i$ to be $x=x'\land 10=10$,
$\Arr(t_i,t'_i)$ to be $\Arr(x'+1,v')$,
$\Sigma_i$ to be $\Emp$
and $S$ to be $\emptyset$.
Hence we have
\begin{align*}
  P&(x'<v' \land x<v \land x'+1=v', \Emp,\\
  &\qquad\{(x=x'\land 10=10\land v=x'+1 \land 11=11,\Emp)\})
  \tag{b1'}
  \\
  &\land
  P(x'<v' \land x<v \land x'+1<v', \Emp,\\
  &\qquad\{(x=x'\land 10=10\land v=x'+1 \land 11=11,\Arr(x'+2,v'))\})
  \tag{b2'}
  \\
  &\land
  \neg(x'<v' \land x<v \land x'+1>v')
  \tag{b3'}    
\end{align*}
by $(\mapsto\mapsto)$ and (empty).
We note that (b2') becomes
$P(x'<v'\land x<v \land x'+1<v', \Emp,\emptyset\})$ by (EmpNEmp)
taking 
$\Pi$ to be $x'<v' \land x<v \land x'+1<v'$,
$\Pi'$ to be $x=x'\land 10=10\land v=x'+1 \land 11=11$,
$\Sigma'$ to be $\Arr(x'+2,v')$,
and $S$ to be $\emptyset$.
Thus we obtain
\begin{align*}
  \bigl((x'<v'& \land x<v \land x'+1=v') \to (x=x'\land 10=10\land v=x'+1 \land 11=11)\bigr)
  \tag{b21}
  \\
  &\land
  \neg(x'<v' \land x<v \land x'+1<v')
  \land
  \neg(x'<v' \land x<v \land x'+1>v')
  \tag{b22}, 
\end{align*}
where (b21) is obtained from (b1') by applying (EmpEmp),
and (b22) is obtained from (b2') and (b3') by applying (empty). 
Finally, by combining (a'), (b21), (b22), and (c'), we have
\begin{align*}
\neg&(x'=v'\land x<v)
\\
&\land
  \bigl((x'<v' \land x<v \land x'+1=v') \to (x=x'\land 10=10\land v=x'+1 \land 11=11)\bigr)
\\
&\land
\neg(x'<v' \land x<v \land x'+1<v')
\land
\neg(x'<v' \land x<v \land x'+1>v')
\\
&\land
  \neg(x'>v'\land x<v)
\end{align*}
as the translation result of
$P(\True,x\mapsto 10 * v\mapsto 11, \{(\True,\Arr(x',v'))\})$. 
\end{exa}

\subsection{Decidability Theorem}

The aim of $P$ is to give an equivalent formula of Presburger arithmetic to a given entailment. 
The correctness property of $P$ is stated as follows. 

\begin{prop}[Correctness of Translation $P$]\label{prop:correct}
If any array atomic formula in $\Sigma_i$ has the form $\Arr(t,t+u)$
such that the term $u$ does not contain $\Vec{y}$, then 
\[
\Pi\land\Tilde\Sigma \models \{\exists \Vec{y_i}(\Pi_i\land\Tilde\Sigma_i)\}_{i\in I}
\quad
\hbox{ iff }
\quad
\models \forall \Vec z\exists \Vec yP(\Pi,\Sigma,\{(\Pi_i,\Sigma_i)\}_{i\in I})
\]
where $\Vec y$ is a sequence of $\Vec{y_i}$ ($i\in I$), and
$\Vec z$ is $\FV(P(\Pi,\Sigma,\{(\Pi_i,\Sigma_i)\}_{i\in I}))\setminus \FV(\Pi,\Sigma,\{\Pi_i\}_{i\in I},\{\Sigma_i\}_{i\in I})$. 
\end{prop}
We note that $\Vec z$ are the fresh variables introduced in the unfolding of $P(\Pi,\Sigma,\{(\Pi_i,\Sigma_i)\}_{i\in I})$. 
The proof of this theorem will be given in the next section. 

The correctness property is shown with the condition described in the theorem. 
This condition avoids a complicated situation for $\Vec y$ and $\Vec z$, 
such that some variables in $\Vec y$ depend on $\Vec z$, and some determine $\Vec z$. 
For example, if we consider $\Arr(1,5)\vdash \exists y_1y_2(\Arr(1,y_1)*y_1+1\mapsto y_2*\Arr(y_1+2,5))$, 
we will have $y_1+1\mapsto z$ during the unfolding of 
$P(\True,\Arr(1,5),\{(\True,\Arr(1,y_1)*y_1+1\mapsto y_2*\Arr(y_1+2,5))\})$. 
Then we have $z=y_2$ after some logical simplification. 
This fact means that $y_2$ depends on $z$, 
and moreover $z$ is indirectly determined by $y_1$. 
The latter case occurs when sizes of array depend on $\Vec y$. 
We need to exclude this situation. 

Finally we have the decidability result for the entailment problem of $\SLAR$ under the condition from the above theorem and the property of sorted entailments (stated in Lemma~\ref{lemma:split_sorted}). 

\begin{thm}[Decidability of Validity Checking of Entailments in $\SLAR$]\label{thm:decidability-slar}
Validity checking of entailments
$\Pi\land\Sigma \vdash \{\exists \Vec{y_i}(\Pi_i\land\Sigma_i)\}_{i\in I}$ in $\SLAR$ 
  is decidable, 
if any array atomic formula in $\Sigma_i$ has the form $\Arr(t,t+u)$ such that the term $u$ does not contain $\Vec{y_i}$. 
\end{thm}

\begin{exa}
An example
$\Arr(x,x) \vdash x \mapsto 0, \exists y(y > 0 \land x \mapsto y)$
satisfies the condition, and its validity is checked in the following way. 

  \begin{itemize}
    \item It is decomposed into several sorted entailments: in this case, it produces one sorted entailment
$\Arr(x,x)^\sim \vdash x \mapsto 0, \exists y(y > 0 \land x \mapsto y)$

\item Compute $P(\True,\Arr(x,x), S_1)$, 
where $S_1$ is $\{(\True,x \mapsto 0),(y>0, x \mapsto y)\}$. 
It becomes
$P(x<x, x\mapsto z*\Arr(x+1,x), S_1) \land P(x=x, x\mapsto z, S_1)$
by (Arr$\mapsto$). Then it becomes 
      \[P(x<x\land x<x+1, \Arr(x+1,x), S_2) \land P(x=x\land x<x+1, \Emp,
      S_2),\]
where $S_2$ is 
$\{(x=x\land z=0,\Emp),(y>0 \land x=x\land z=y,\Emp)\}$. 
The former conjunct becomes $P(x<x \land x<x+1,\Arr(x+1,x),\emptyset)$ by (NEmpEmp), 
then, by (empty), it becomes 
$\neg(x<x \land x<x+1 \land x+1\le x)$. 
The latter conjunct becomes 
$x=x\land x<x+1 \to (x=x \land z=0)\vee(y>0\land x=x\land z=y)$, 
which is equivalent to $x<x+1 \to z=0\vee(y>0\land z=y)$, 

\item Check the validity of the formula
$\forall xz\exists yP(\True,\Arr(x,x), S_1)$, 
which is equivalent to 
$\forall xz\exists y(\neg(x<x \land x<x+1 \land x+1\le x) \land (z=0\vee (y>0\land z=y)))$. 
Finally the procedure answers ``valid'', since the produced Presburger formula is valid. 
  \end{itemize}
\end{exa}

\subsection{Other Systems of Symbolic Heaps with Arrays}

Other known systems of symbolic heaps with arrays are only
the system given in Brotherston et al.~\cite{Brotherston17}.
They gave a different condition for decidability of the entailment problem for the same symbolic-heap system. 
Their condition disallows existential variables in $u$
for each points-to atomic formula $t\mapsto u$ in the succedent of an entailment.
In order to clarify the difference between our condition and their condition,
we consider the following entailments:

\begin{enumerate}[(i)]
  \item
$\Arr(x,x) \vdash x \mapsto 0, \exists y(y > 0 \land x \mapsto y)$, 
\item
$\Arr(1,5) \vdash \exists y,y'(\Arr(y,y+1) * \Arr(y',y'+2))$, 
\item
$\Arr(1,5) \vdash \exists y(\Arr(1,1+y) * \Arr(2+y,5))$, 
\item
$\Arr(1,5) \vdash \exists y,y'(\Arr(1,1+y) * 2+y \mapsto y' * \Arr(3+y,5))$. 
\end{enumerate}

The entailment (i) can be decided by our decision procedure, but it cannot be decided by their procedure. 
The entailment (ii) is decided by both theirs and ours. 
The entailment (iii) is decided by theirs, but it does not satisfy our condition. 
The entailment (iv) is decided by neither theirs nor ours. 

Our system and the system in \cite{Brotherston17} have the
same purpose, namely, analysis/verification of memory safety.
Basically their target programming language and assertion language are
the same as ours given in this section.
These entailment checkers are essentially used for deciding the side condition
of the consequence rule.
As explained above, ours and theirs have different restrictions for decidability.
Hence the class of programs is the same for
ours and theirs, but some triples can be proved only by ours and
other triples can be proved only by theirs,
according to the shape of assertions.
In this sense both our system and their system
have advantage and disadvantage.

%% file: 4.tex
\section{Correctness of Decision Procedure}\label{sec:correctness}

This section shows correctness of our decision procedure. 
We first show the basic property of sorted entailments. 

\subsection{Correctness of Translation}

This subsection shows correctness of the translation $P$. 
The main difficulty for showing correctness is how to handle the new variables (denoted by $z$)
that are introduced during the unfolding $P$. 
In order to do this, we temporarily extend our language with new terms, denoted by $[t]$. 
A term $[t]$ means the value at the address $t$, that is, 
it is interpreted to $h(s(t))$ under $(s,h)$.
We will use this notation instead of $z$, since $z$ must appear in the form $t \mapsto z$  during unfolding $P$, and this $t$ is unique for $z$. 
Note that both $s$ and $h$ are necessary for interpreting 
a formula of the extended language even 
if it is a pure formula. 

In this extended language, we temporarily introduce a variant $P'$ of $P$
so that we use $[t]$ instead of $z$, 
which is defined in the same way as $P$ except
\begin{align*}
P'(\Pi,\Arr(t,t') * \Sigma,S) 
\eqDef
&
P'(\Pi \land t'=t, t \mapsto [t] * \Sigma,S)
\\
&\land
P'(\Pi \land t'>t, t \mapsto [t] * \Arr(t+1,t') * \Sigma,S), 
\end{align*}
when $(\Pi'', t'' \mapsto u'' * \Sigma'') \in S$.
Note that $P'$ never introduces any new variables.

We will introduce some notations. 
Let $S$ be $\{(\Pi,\Sigma)\}_{i\in I}$. 
Then we write $\Tilde S$ for $\{ \Pi_i \land \Tilde{\Sigma_i} \}_{i\in I}$.
We write $\Dom(s,\Sigma)$ for 
the set of addresses used by $\Sigma$ under $s$, that is, 
it is inductively defined as follows: 
$\Dom(s,\Emp) = \emptyset$, 
$\Dom(s,\Emp * \Sigma_1) = \Dom(s,\Sigma_1)$, 
$\Dom(s,t\mapsto u*\Sigma_1) = \{s(t)\}\cup\Dom(s,\Sigma_1)$, and 
$\Dom(s,\Arr(t,u)*\Sigma_1) = \{s(t),\ldots,s(u)\}\cup\Dom(s,\Sigma_1)$ if $s(t) \le s(u)$. 

The next lemma clarifies the connections between entailments, $P$, and $P'$. 

\begin{lem}\label{lemma:correct}
{\rm (1)}
Assume $s,h \models \hat \Pi \land \hat \Sigma$.
Suppose $P'(\Pi,\Sigma,S)$ appears in the unfolding of $P'(\hat\Pi, \hat\Sigma, \hat S)$.
Then

$
s,h|_{\Dom(s,\Sigma)} \models P'(\Pi,\Sigma,S)
$
iff
$
s,h|_{\Dom(s,\Sigma)} \models \Pi \land \Sorted(\Sigma) \imp \Lor\Tilde S.
$

{\rm (2)}
$
\forall sh(s,h \models \Pi\land\Tilde\Sigma \imp s,h \models \exists \Vec yP'(\Pi,\Sigma,S))
$
iff
$
\Pi\land\Tilde\Sigma \models \exists \Vec y\Lor\Tilde S.
$

{\rm (3)} $\models \neg(\Pi \land \Sorted(\Sigma)) \imp P(\Pi,\Sigma,S)$.

{\rm (4)}
  Let $\Vec z$ be fresh variables introduced during the calculation of $P(\Pi,\Sigma,S)$.
  Then
$
\forall sh(s,h \models \Pi \land \Tilde\Sigma \imp s,h \models 
\forall \Vec z\exists \Vec yP(\Pi,\Sigma,S))
$
iff
$
\models \forall \Vec z\exists \Vec yP(\Pi,\Sigma,S).
$
\end{lem}
\noindent\textbf{Proof of Lemma~\ref{lemma:correct}~(1)}. \quad
This is shown by induction on the steps $\eqDef$.
Consider cases according to the definition of $P'$.

{\bfseries Case~1} ($\mapsto\mapsto$-case):
\begin{align*}
P'&(\Pi,t \mapsto u * \Sigma,\{( \Pi_i, t_i \mapsto u_i * \Sigma_i) \}_{i \in I }) 
\\
&\eqDef P'(\Pi \land t < \Sigma,\Sigma, \{( \Pi_i \land t=t_i \land u=u_i \land t_i < \Sigma_i,\Sigma_i) \}_{i \in I}).
\end{align*}

Let $h_1=h|_{\Dom(s,t \mapsto u * \Sigma)}$,
$h_2=h|_{\Dom(s,\Sigma)}$. 
Then $h_1=\{(s(t),h(s(t)))\}+h_2$. 
It is enough to show
\begin{equation}\label{eq_1}
s,h_1 \models 
\Pi \land \Sorted(t \mapsto u * \Sigma) \imp \Lor_{i\in I} \Pi_i\land (t_i \mapsto u_i * \Sigma_i)^\sim
\end{equation}
iff
\begin{equation}\label{eq_2}
s,h_2 \models 
\Pi \land t < \Sigma \land \Sorted(\Sigma) \imp \Lor_{i\in I} \Pi_i\land t=t_i \land u=u_i \land t_i < \Sigma_i \land \Sigma_i^\sim.
\end{equation}

The only-if part.
Assume (\ref{eq_1}) and 
the antecedent of (\ref{eq_2}). 
Then the antecedent of (\ref{eq_1}) holds, since it is equivalent to the antecedent of (\ref{eq_2}). 
Then the succedent of (\ref{eq_1}) is true for $s,h_1$.
Hence the succedent of (\ref{eq_2}) is true for $s,h_2$.

The if part.
Assume (\ref{eq_2}) and 
the antecedent of (\ref{eq_1}).
Then the antecedent of (\ref{eq_2}) holds, since it is equivalent to the antecedent of (\ref{eq_1}).
Then the succedent of (\ref{eq_2}) is true for $s,h_2$.
Hence the succedent of (\ref{eq_1}) is true for $s,h_1$.

{\bfseries Case~2} ({\bf Arr}$\mapsto$-case):
\begin{align*}
P'(\Pi,\Arr(t,t') * \Sigma,S) 
\eqDef&
P'(\Pi \land t'=t \land t \le \Sigma,\Sigma,S)
\\
&\land
P'(\Pi \land t'>t,t \mapsto [t] * \Arr(t+1,t') * \Sigma,S). 
\end{align*}

Let $h_3 = h|_{\Dom(s,\Arr(t,t') * \Sigma)}$,
$h_4=h|_{\Dom(s,t\mapsto [t] * \Sigma)}$, and 
$h_5=h|_{\Dom(s,t \mapsto [t] * \Arr(t+1,t') * \Sigma)}$. 
It is enough to show 
\begin{equation}\label{eq_3}
s,h_3 \models 
\Pi \land \Sorted(\Arr(t,t') * \Sigma) \imp \Lor \Tilde S
\end{equation}
is equivalent to the conjunction of the following two clauses: 
\begin{equation}
    s,h_4 \models
    \Pi \land t'=t \land t < \Sigma \land \Sorted(\Sigma) \imp \Lor \Tilde S
  \label{eq_4}
\end{equation}
and
\begin{equation}
    s,h_5 \models
    \Pi \land t'>t \land \Sorted(t \mapsto [t] * \Arr(t+1,t') * \Sigma)
    \imp \Lor \Tilde S.
  \label{eq_5}
\end{equation}

{\bfseries Case~2.1:} the case of $s(t)=s(t')$.

We note that $h_3 = h_4$. 
The antecedent of (\ref{eq_4}) is equivalent to
the antecedent of (\ref{eq_3}). 
(\ref{eq_5}) is true since $s(t)=s(t')$.
Hence (\ref{eq_3}) and $(\ref{eq_4})\land(\ref{eq_5})$ are equivalent.

{\bfseries Case~2.2:} the case of $s(t')>s(t)$. 

We note that $h_3=h_5$.
The antecedent of (\ref{eq_5}) is equivalent to
the antecedent of (\ref{eq_3}).
(\ref{eq_4}) is true since $s(t')>s(t)$. 
Hence (\ref{eq_3}) and $(\ref{eq_4})\land(\ref{eq_5})$ are equivalent.

{\bfseries Case~3} ($\mapsto${\bf Arr}-case):
\begin{align*}
P'(\Pi,&t \mapsto u * \Sigma,\{(\Pi_i,\Arr(t_i,t_i') * \Sigma_i)\} \cup S)
\\
\eqDef&
P'(\Pi \land t_i'=t_i,t \mapsto u * \Sigma,\{(\Pi_i,t_i \mapsto u * \Sigma_i)\} \cup S)
\\ 
&\land
P'(\Pi \land t_i'>t_i,t \mapsto u * \Sigma,\{(\Pi_i,t_i \mapsto u * \Arr(t_i+1,t_i') * \Sigma_i)\} \cup S)
\\
&\land
P'(\Pi \land t_i'<t_i,t \mapsto u * \Sigma,S)
\end{align*}
This case is proved by showing the following claim, 
which is shown similarly to the claim of Case~2.
Let $h' = h|_{\Dom(s,t \mapsto u * \Sigma)}$.
Let $\Sorted_L$ be an abbreviation of $\Sorted(t \mapsto u * \Sigma)$. 
Then 
\[
s,h'\models 
\Pi \land \Sorted_L \imp 
\Pi_i \land (\Arr(t_i,t_i') * \Sigma_i)^\sim \lor \Lor \Tilde S
\]
is equivalent to the conjunction of the following three clauses: 

\begin{align*}
  s,h' & \models
\Pi \land t_i'=t_i\land \Sorted_L
\imp 
\Pi_i \land (t_i \mapsto u * \Sigma_i)^\sim \lor \Lor \Tilde S
\\
  s,h' & \models
\Pi \land t_i'>t_i\land \Sorted_L
\imp \Pi_i \land (t_i \mapsto u * \Arr(t_i+1,t_i') * \Sigma_i)^\sim \lor \Lor \Tilde S,
\\
  s,h' & \models
\Pi \land t_i'<t_i\land \Sorted_L \imp \Lor \Tilde S. 
\end{align*}

{\bfseries Case~4} ({\bf (ArrArr)}-case): 
Consider that
$P'(\Pi,\Arr(t,t') * \Sigma,\{( \Pi_i,\Arr(t_i,t_i') * \Sigma_i)\}_{i \in I})$
is defined by the conjunction of 
  \[\hspace{-.5cm}P'\left(%
\begin{array}{l}
    \Pi \land
    m=m_{I'}\land m<m_{I\setminus I'}
    \land t \le t' \land t' < \Sigma \land 0 < t, \Sigma,
    \\
    \{( \Pi_i \land t = t_i \land t'_i < \Sigma_i, \Sigma_i)\}_{i \in I'}
    \cup
    \{( \Pi_i \land t = t_i,\Arr(t_i+m+1,t_i') * \Sigma_i)\}_{i \in I\setminus I'}
  \end{array}%
\right)\]
for all $I'\subseteq I$ and 
  \[\hspace{-.5cm}P'\left(%
\begin{array}{l}
    \Pi \land m'<m \land m'=m_{I'} \land m'<m_{I\setminus I'} \land 0 < t,
    \Arr(t+m'+1,t') * \Sigma,
    \\
    \{( \Pi_i \land t = t_i \land t'_i < \Sigma_i, \Sigma_i )\}_{i \in I'}
    \cup
    \{( \Pi_i \land t = t_i,\Arr(t_i+m'+1,t_i') * \Sigma_i )\}_{i \in I\setminus I'}
  \end{array}%
\right)\]
for all $I' \subseteq I$ with $I' \neq \emptyset$, 
where $m$, $m_i$, and $m'$ are abbreviations of 
$t'-t$, $t_i'-t_i$, and $m_{\min I'}$, respectively.

Let 
$h_6 = h|_{\Dom(s,\Arr(t,t') * \Sigma)}$, 
$h_7 = h|_{\Dom(s,\Sigma)}$, and 
$h_8 = h|_{\Dom(s,\Arr(t+m'+1,t') * \Sigma)}$. 
It is enough to show 
\begin{equation}
\label{eq_6}
s,h_6 \models 
\Pi \land \Sorted(\Arr(t,t') * \Sigma) \imp \Lor_{i\in I} \Pi_i \land (\Arr(t_i,t_i') * \Sigma_i))^\sim
\end{equation}
is equivalent to the conjunction of the following
\begin{align}
    s,h_7& \models
    \Pi\land m=m_{I'} \land m<m_{I\setminus I'} \land 0 < t
    \land t \le t' \land t' < \Sigma \land \Sorted(\Sigma)\notag
  \\
    &\imp \Lor_{i\in I'}
    \Pi_i \land t = t' \land t'_i < \Sigma_i\land\Sigma_i^\sim\notag
  \\
    &\qquad \vee\
    \Lor_{i\in I\setminus I'}\Pi_i \land t = t' \land (\Arr(t_i+m'+1,t_i') * \Sigma_i)^\sim
  \label{eq_7}
\end{align}
for any $I' \subseteq I$, 
and
\begin{align}
  \hspace{-1cm}
    s, h_8 & \models
    \Pi\land m'<m \land m'=m_{I'}
    \land m'<m_{I\setminus I'} \land 0 < t
    \land \Sorted(\Arr(t+m'+1,t') * \Sigma)\notag
  \\
    &\imp \Lor_{i \in I'}
    \Pi_i \land t = t' \land t'_i < \Sigma_i\land\Sigma_i^\sim\notag
  \\
    &\qquad \vee\
    \Lor_{i \in I\setminus I'}\Pi_i \land t = t' \land(\Arr(t_i+m'+1,t_i') * \Sigma_i)^\sim
  \label{eq_8}
\end{align}
for any $I' \subseteq I$ with $I' \neq \emptyset$.

{\bfseries Case~4.1:} the case of $s \models m=m_{I'} \land m<m_{I\setminus I'}$ for some $I' \subseteq I$. 

The antecedent of the conjunct of the form of (\ref{eq_7}) with $I'$ being this $I'$ satisfies the case condition.
The conjunct of the form (\ref{eq_7}) with $I'$ begin not this $I'$ and 
the conjuncts of the form (\ref{eq_8}) are true,
because their antecedents are false by the case condition.

Now we show the only-if part and the if part 
by using $h_6= h|_{\{s(t),s(t+1),\ldots,s(t')\}} + h_7$.

The only-if part: 
Suppose (\ref{eq_6}) is true. 
By the above observation it is enough to consider the conjunct of the form (\ref{eq_7}) with this $I'$. 
Assume the antecedent of the conjunct is true for $s,h_7$. 
Then its succedent is also true for $s,h_7$ since the succedent of (\ref{eq_6}) is true for $s,h_6$. 
Therefore the form of (\ref{eq_7}) with $I'$ being this $I'$ holds. 

The if part:
Suppose that the forms of (\ref{eq_7}) for any $I'\subseteq I$ hold, and
the forms of (\ref{eq_8}) for any $I'$ such that $\emptyset \neq I' \subseteq I$ hold.
Assume that the antecedent of (\ref{eq_6}) is true for $s,h_6$.
Then the succedent of the form (\ref{eq_7}) with $I'$ being this $I'$ holds for $s,h_7$ since its antecedent is true. 
Hence the succedent of (\ref{eq_6}) is true for $s,h_6$.
Therefore (\ref{eq_6}) holds. 

{\bfseries Case~4.2:} the case of  
$s \models m'<m \land m'=m_{I'} \land m'<m_{I\setminus I'}$ for some $I' \subseteq I$.
It is similar to Case 4.1 by using
$h_6 = h|_{\{s(t), \ldots, s(t+m') \}} + h_8$. 

\noindent\textbf{Proof of Lemma~\ref{lemma:correct}~(2)}.\quad
We first show the only-if part. 
Assume the left-hand side of the claim and 
$s,h \models \Pi \land \Tilde\Sigma$.
By the left-hand side, we obtain 
$s,h \models \exists \Vec yP'(\Pi,\Sigma,S)$.
Hence we have $s'$ such that $s',h \models P'(\Pi,\Sigma,S)$.
By (1),
$s',h \models \Pi \land \Sorted(\Sigma) \imp \Lor S$.
Thus $s',h \models \Lor S$. 
Finally we have $s,h \models \exists \Vec y\Lor S$.

Next we show the if part. 
Fix $s,h$.
Assume the right-hand side of the claim and 
$s,h \models \Pi \land \Tilde\Sigma$.
By the right-hand side, $s,h \models \exists \Vec y\Lor \Tilde S$ holds.
Hence we have $s'$ such that $s',h \models \Lor S$.
Then we obtain $s',h \models \Pi \land \Sorted(\Sigma) \imp \Lor \Tilde S$.
By (1), $s',h \models P'(\Pi,\Sigma,S)$ holds.
Finally we have $s,h \models \exists \Vec yP'(\Pi,\Sigma,S)$.

\noindent\textbf{Proof of Lemma~\ref{lemma:correct}~(3)}.\quad
It is shown by induction on the steps of $\eqDef$. 
Consider cases according to the definition of $P$. 

{\bfseries Case~1} (({\bf EmpL})-case):
By the induction hypothesis we have
$\models \neg(\Pi\land\Sorted(\Sigma)) \imp P(\Pi,\Sigma,S)$. 
This is equivalent to
$\models \neg(\Pi\land\Sorted(\Emp * \Sigma)) \imp P(\Pi,\Emp * \Sigma,S)$. 

The other cases ({\bf EmpR}), ({\bf EmpNEmp}), ({\bf NEmpEmp}), and ($\mapsto\mapsto$)
are shown in a similar way. We will consider the remaining cases. 

{\bfseries Case~2} (({\bf EmpEmp})-case): 
This case is easily shown, since
$\models \neg(\Pi\land\Sorted(\Emp)) \imp (\Pi \imp \bigvee_{i\in I}\Pi_i)$ trivially holds. 

{\bfseries Case~3} (({\bf empty})-case):
This case is easily shown, since
$\models \neg(\Pi\land\Sorted(\Emp)) \imp \neg(\Pi\land\Sorted(\Emp))$ trivially holds. 

{\bfseries Case~4} (($\mapsto${\bf Arr})-case):
By the induction hypothesis, we have 
\[
\models \neg(\Pi \land t'_i=t_i\land\Sorted(t\mapsto u*\Sigma)) \imp P^{(=)}, 
\]
where $P^{(=)}$ is an abbreviation of $P(\Pi \land t'_i=t_i,t\mapsto u*\Sigma,\{(\Pi_i,t_i \mapsto u * \Sigma_i)\} \cup S)$, 
\[
\models \neg(\Pi \land t'_i>t_i\land\Sorted(t\mapsto u*\Sigma)) \imp P^{(>)}, 
\]
where
$P^{(>)}$ is an abbreviation of $P(\Pi \land t'_i>t_i,t\mapsto u*\Sigma,\{(\Pi_i,t_i \mapsto u * \Arr(t_i+1,t_i') * \Sigma_i)\} \cup S)$, and 
\[
\models \neg(\Pi \land t'_i<t_i\land\Sorted(t\mapsto u*\Sigma)) \imp P^{(<)}, 
\]
where $P^{(<)}$ is an abbreviation of $P(\Pi \land t'_i<t_i,t\mapsto u*\Sigma,S)$. 
In order to show the current case, it is enough to show
$\Pi \land t'_i=t_i\land\Sorted(t\mapsto u*\Sigma) \imp \Pi \land \Sorted(t \mapsto u * \Sigma)$,
$\Pi \land t'_i>t_i\land\Sorted(t\mapsto u*\Sigma) \imp \Pi \land \Sorted(t \mapsto u * \Sigma)$, and
$\Pi \land t'_i<t_i\land\Sorted(t\mapsto u*\Sigma) \imp \Pi \land \Sorted(t \mapsto u * \Sigma)$. 
They trivially hold by comparing conjuncts. 

{\bfseries Case~5} (({\bf Arr}$\mapsto$)-case):
By the induction hypothesis, we have
\[
\models \neg(\Pi \land t'>t \land \Sorted(t\mapsto z*\Arr(t+1,t')*\Sigma)) \imp P^{(>)},
\]
where $P^{(>)}$ is an abbreviation of $P(\Pi \land t' > t,t \mapsto z * \Arr(t+1,t') * \Sigma,S)$, and
\[
\models \neg(\Pi \land t'=t \land \Sorted(t\mapsto z'*\Sigma)) \imp P^{(=)},
\]
where $P^{(=)}$ is an abbreviation of $P(\Pi \land t' = t,t \mapsto z' * \Sigma,S)$.
In order to show the current case, it is enough to show
$\Pi \land t'>t \land \Sorted(t\mapsto z*\Arr(t+1,t')*\Sigma) \imp \Pi \land \Sorted(\Arr(t,t')*\Sigma)$
and
$\Pi \land t'=t \land \Sorted(t\mapsto z'*\Sigma) \imp \Pi \land \Sorted(\Arr(t,t')*\Sigma)$.
They are immediately shown by the definition of $\Sorted$.

{\bfseries Case~6} (({\bf ArrArr})-case):
For $I' \subseteq I$, we abbreviate
$m = m_{I'} \land m < m_{I\setminus I'}$
and
$m'<m \land m' = m_{I'} \land m'<m_{I\setminus I'}$ by 
$\Pi^{(=)}_{I'}$ and $\Pi^{(<)}_{I'}$, respectively.
We write $S^{(=)}_{I'}$ and $S^{(<)}_{I'}$ for the third arguments of the first and the second $P$ for $I'$
in the right-hand side of (ArrArr), respectively. 
We also write $\Arr$ for $\Arr(t+m'+1,t')$. 
By the induction hypothesis, we obtain the forms of 
\begin{align*}
  \models
  \neg(\Pi\land\Pi^{(=)}_{I'} \land
  t \le t' \land
  &
  t' < \Sigma \land
  0 < t \land
  \Sorted(\Sigma))
  \\
  &\imp
  P(\Pi\land\Pi^{(=)}_{I'} \land t \le t' \land t' < \Sigma \land 0 < t, \Sigma, S^{(=)}_{I'})
\end{align*}
for each $I'\subseteq I$, and
\begin{align*}
  \models
  \neg(\Pi\land\Pi^{(<)}_{I'}\land
  0 < t \land
  \Sorted(\Arr * \Sigma))
  \imp
  P(\Pi\land\Pi^{(<)}_{I'}\land 0 < t, \Arr * \Sigma,S^{(<)}_{I'})
\end{align*}
for each $I' \subseteq I$ with $I' \neq \emptyset$. 
In order to show the current case, it is enough to show
$\Pi\land\Pi^{(=)}_{I'} \land t \le t' \land t' < \Sigma \land 0 < t \land \Sorted(\Sigma) \imp \Pi \land \Sorted(\Arr(t,t')*\Sigma)$ for each $I'\subseteq I$, 
and
$\Pi\land\Pi^{(<)}_{I'}\land 0 < t \land \Sorted(\Arr * \Sigma) \imp \Pi \land \Sorted(\Arr(t,t')*\Sigma)$
for each $I' \subseteq I$ with $I' \neq \emptyset$. 
They are immediately shown by the definition of $\Sorted$.

\noindent\textbf{Proof of Lemma~\ref{lemma:correct}~(4)}.\quad
We note that 
$\forall sh(s,h \models \Pi \land \Tilde\Sigma \imp s \models \forall \Vec z\exists \Vec yP(\Pi,\Sigma,S))$
is equivalent to
$\forall s(\exists h(s,h \models \Pi \land \Tilde\Sigma) \imp s \models \forall \Vec z\exists \Vec yP(\Pi,\Sigma,S))$.
Moreover it is equivalent to
$\forall s(s \models \Pi \land \Sorted(\Sigma) \imp s \models \forall \Vec z\exists \Vec yP(\Pi,\Sigma,S))$.
By (3), we have 
$\models \neg(\Pi \land \Sorted(\Sigma)) \imp \forall \Vec z\exists \Vec yP(\Pi,\Sigma,S)$, since 
$\Pi \land \Sorted(\Sigma)$ does not contain $\Vec z$,
which are fresh variables 
introduced during the calculation of $P(\Pi,\Sigma,S)$.
Hence we obtain
$\models
(\Pi \land \Sorted(\Sigma)) \vee \neg(\Pi \land \Sorted(\Sigma))
\imp
\forall \Vec z\exists \Vec yP(\Pi,\Sigma,S)$. 
This is equivalent to
$\models \forall \Vec z\exists \Vec yP(\Pi,\Sigma,S)$.
\qed

\subsection{Decidability Proof}

This subsection proves the correctness of the decision procedure. 
\\
\noindent\textbf{Proof of Proposition~\ref{prop:correct}}
\quad
Let $S$ be $\{(\Pi_i,\Sigma_i)\}_{i\in I}$.
Then the left-hand side is equivalent to
$\Pi \land \Tilde\Sigma \models \exists\Vec y\Lor \Tilde S$.
Moreover, by Lemma~\ref{lemma:correct}~(2), it is equivalent to
\begin{equation}\label{eq_Z}
\forall sh(s,h \models \Pi \land \Tilde\Sigma \imp s,h \models \exists \Vec yP'(\Pi,\Sigma,S)). 
\end{equation}
By Lemma~\ref{lemma:correct}~(4), the right-hand side is equivalent to 
\begin{equation}\label{eq_A}
\forall sh(s,h \models \Pi \land \Tilde\Sigma \imp s \models \forall \Vec z\exists \Vec yP(\Pi,\Sigma,S)). 
\end{equation}
Now we will show the equivalence of (\ref{eq_Z}) and (\ref{eq_A}). 
Here we assume $[t_1],\ldots,[t_n]$ appear in $P'(\Pi,\Sigma,S)$
and $s\models t_1 < \ldots < t_n$,
we let $\Vec z = z_1,\ldots,z_n$.

Recall that our condition requires that
sizes of arrays in the succedent do not depend on existential variables.
We note that, under the condition of Theorem~\ref{thm:decidability-slar}, 
each $t$ of $t \mapsto u$ or $\Arr(t,t')$ in the second argument of $P'$ during the unfolding of $P'$ does not contain any existential variables. 
By this fact, we can see that each term $[t]$ does not contain existential variables,
since it first appears as $t \mapsto [t]$ in the second argument of $P'$ during the unfolding of $P'$. 
So we can obtain $P'(\Pi,\Sigma,S) = P(\Pi,\Sigma,S)[\Vec z:=[\Vec t]]$. 
Hence (\ref{eq_Z}) is obtained from (\ref{eq_A}) by taking $z_i$ to be $[t_i]$ for $1 \le i \le n$.

We show the inverse direction. Assume (\ref{eq_Z}). 
Fix $s$ and $h$ such that $s,h \models \Pi\land\Tilde\Sigma$. 
We will show 
$s \models \forall \Vec z\exists \Vec yP(\Pi,\Sigma,S)$. 
Take $\Vec a$ for $\Vec z$.
Let $s'$ be $s[\Vec z:=\Vec a]$.
We claim that $s(t_j) \in \Dom(h)$ ($j = 1,\ldots,n$), 
since each $t_j$ appears as an address of an array atomic formula in $\Sigma$. 
Define $h'$ by $\Dom(h') = \Dom(h)$ and $h'(m) = a_j$ if $m = s(t_j)$ ($j = 1,\ldots,n$), otherwise $h'(m) = h(m)$. 
Then we have $s,h' \models \Pi \land \Tilde\Sigma$. 
By (\ref{eq_Z}), we obtain $s,h' \models \exists \Vec yP'(\Pi,\Sigma,S)$. 
Hence $s'\models \exists\Vec yP(\Pi,\Sigma,S)$ holds, 
since $s',h'\models [\Vec t] = \Vec z$. 
Therefore we have $s\models \forall\Vec z\exists\Vec yP(\Pi,\Sigma,S)$. 
Thus we obtain (\ref{eq_A}). 
\qed

Hence we obtained the decidability result of $\SLAR$ stated in Theorem~\ref{thm:decidability-slar}. 

%% file: 5.tex
\section{Separation Logic with Arrays and Lists}

From this section, 
we will show the second result of this paper: 
the decidability of validity checking for QF entailments of symbolic heap system with array and list predicates. 
The decidability result of the previous section will be used in that of the second result. 

We start from the separation logic $\SLGL$, which is obtained from $\SLG$ 
by assuming the point-to predicate is ternary
and 
adding the singly-linked list predicate $\Ls(\wild,\wild)$ 
and the doubly-linked list predicate $\Dll(\wild,\wild,\wild,\wild)$. 
Then we define the symbolic heap system $\SLLA$ which is a fragment of $\SLGL$. 

\subsection{Syntax of $\SLGL$ and $\SLLA$}

The terms of $\SLGL$ are same as those of $\SLG$. 
The formulas (denoted by $F$) of $\SLGL$ are defined as follows: 
\\

$F ::= t = t\ |\ F \land F\ |\ \neg F\ |\ \exists xF\ |\ \Emp \ |\ F*F
\ |\ t \mapsto (t,t)\ |\ \Arr(t,t)\ |\ \Ls(t,t) \ |\ \Dll(t,t,t,t)$.
\\

The notations for $\SLG$ mentioned in Section~\ref{sect:sla} are also used for $\SLGL$. 

We call the singly-linked and doubly-linked list predicates {\em list predicates}. 
We also call a formula {\em list-free} if it does not contain any list predicates.

The list predicates are inductively defined predicates and they are introduced with the following definitions clauses: 
\begin{align*}
\Ls(x,y) 
&::= 
(x = y \land \Emp) \vee \exists zw(x\mapsto (z,w) * \Ls(z,y)),
\\
\Dll(x_1,y_1,x_2,y_2)
&::=
(x_1 = y_1 \land x_2 = y_2 \land \Emp)
\vee
\exists x(x_1\mapsto (x,y_2) * \Dll(x,y_1,x_2,x_1)).
\end{align*}

We define the $k$-times unfolding of the list predicates. 
\begin{defi}
For $k\ge 0$, $\Ls^k(t,u)$ and $\Dll^k(t,u,v,w)$ are inductively defined by: 

$\Ls^0(t,u) \eqDef 0\neq 0 \land \Emp$, 

$\Ls^{k+1}(t,u) \eqDef (t=u\land\Emp) \vee \exists zw(t \Pto (z,w) * \Ls^k(z,u))$,

$\Dll^0(t,u,v,w) \eqDef 0\neq 0 \land \Emp$, and 

$\Dll^{k+1}(t,u,v,w) \eqDef
(t = u \land v = w \land \Emp)
\vee
\exists z(t \Pto (z,w) * \Dll^k(z,u,v,t))$.
\end{defi}

For each formula $F$ of $\SLGL$ and heap model $(s,h)$, we define $s,h\models F$ extending that of $\SLG$ with 
\begin{align*}
s,h&\models \Ls(t,u)
\iffDef
s,h\models \Ls^k(t,u)
\hbox{ for some $k\ge 0$, and}
\\
s,h&\models \Dll(t,u,v,w)
\iffDef
s,h\models \Dll^k(t,u,v,w)
\hbox{ for some $k\ge 0$.}
\end{align*}

  We note that there is no heap model $(s,h)$ such that $s,h \models \Ls^0(t,u)$. 
  Intuitively a heap model $(s,h)$ that satisfies $s,h\models \Ls^{k+1}(t,u)$ contains a singly-linked list
  of length $k$ starting from $s(t)$ and ends at $s(u)$.
  This situation is depicted as follows: 
\[
\begin{tikzpicture}[every node/.style = {shape=rectangle,text width=6pt,text height=5pt,text centered}]
\mkarrow{0}{$s(t)$}
\mkcell{1}{}{}
\mkarrow{1}{}
\mkdots{2}{}
\mksarrow{2}{}
\mkcell{3}{}{}
\mkarrow{3}{$s(u)$}
\end{tikzpicture}
\]

  We also note that there is no heap model $(s,h)$ such that $s,h \models \Dll^0(t,u,v,w)$. 
  A heap model $(s,h)$ that satisfies $s,h\models \Dll^{k+1}(t,u,v,w)$ contains
  a doubly-linked list of length $k$ 
  with a forward-link starting from $s(t)$ and ending at $s(u)$, and
  a backward-link starting from $s(v)$ and ending at $s(w)$.
  This situation is depicted as follows:
\[
\begin{tikzpicture}[every node/.style = {shape=rectangle,text width=6pt,text height=5pt,text centered}]
\mksarrow{0}{$s(t)$}
\mkarrowr{0}{$s(w)$}
\mkcell{1}{}{}
\mkarrow{1}{}
\mksarrowr{1}{}
\mkdots{2}{}
\mksarrow{2}{}
\mkarrowr{2}{}
\mkcell{3}{}{}
\mkarrow{3}{$s(u)$}
\mksarrowr{3}{$s(v)$}
\end{tikzpicture}
\]

The notation $F \models \Vec{F}$ is also defined in a similar way to $\SLAR$. 

Formulas of $\SLLA$ are QF symbolic heaps (denoted by $\varphi$) of the form $\Pi\land\Sigma$, 
where its pure part $\Pi$ is the same as that of $\SLAR$ and its spatial part $\Sigma$ is defined by
\[
\Sigma ::= \Emp \mid t \mapsto (t,t) \mid \Arr(t,t) \mid \Ls(t,t) \mid \Dll(t,t,t,t) \mid \Sigma * \Sigma. 
\]

We use notations $\Pi_\varphi$ and $\Sigma_\varphi$ that mean the pure part
and the spatial part of $\varphi$, respectively.

Entailments of $\SLLA$ are QF entailments of the form
$\varphi \vdash \{\varphi_i \ |\ i \in I\}$. 
The validity of an entailment is defined in a similar way to $\SLAR$.

%% file: 6.tex
\section{Unroll Collapse}\label{sect:unroll}

This section shows the unroll collapse properties
for the singly-linked and doubly-linked list predicates. 
In the decision procedure for $\SLLA$, these properties will be used 
for eliminating the list predicates in the antecedent of a given entailment. 

The unroll collapse property for the singly-linked list predicate is given in the next proposition.
The key idea of its proof is to replace the list $\Ls(t,u)$ by some list
$\Ls(t,u)$ of length 2 in the antecedent of (1), in order to show (1).
Then (2) applies to obtain the succedent. Then we can show that the
list of length 2 is contained in some list $\Ls(p,q)$ of the
succedent.  By replacing the list of length 2 by $\Ls(t,u)$ of arbitrary
length, we obtain the succedent of (1), since the difference by the
replacement is absorbed by $\Ls(p,q)$.


\begin{prop}[Unroll Collapse for $\Ls$]\label{prop:unroll-ls}
The following clauses are equivalent:

(1)\quad 
$\Ls(t,u) * \phi \models \Vec{\psi}$,

(2)\quad 
    $t = u \land \phi \models \Vec{\psi}$
    and
    $t \Pto (z,y_1) * z \Pto (u,y_2) * \phi \models \Vec{\psi}$,
    where $z,y_1,y_2$ are fresh.
\end{prop}
\begin{proof}[Proof of Proposition~\ref{prop:unroll-ls}]
From (1) to (2) is trivial.
We consider the inverse direction. 
Assume $s,h\models \Ls(t,u) * \phi$.
We will show $s,h\models \Vec{\psi}$.

By the assumption, there are $h_1$ and $h_2$ such that 
$s,h_1\models \Ls(t,u)$, $s,h_2\models \phi$, and $h = h_1+h_2$. 
Hence $s,h_1\models \Ls^n(t,u)$ for some $n$ (take the smallest one). 

We will show $s,h\models \Vec{\psi}$. 
In the case of $n = 0$, it is not the case since $\Ls^0(t,u)$ is unsatisfiable. 
In the case of $n = 1$ or $n = 3$, we have $s,h\models \Vec{\psi}$ by (2). 

We consider the other cases. 
Let $d$ be the second component of $h_1(s(t))$. 
Then the current situation of $h_1$ and $h_2$ is depicted as follows: 
\begin{center}
\begin{tabular}{lccc}
Formula&$\Ls(t,u)$&$*$&$\phi$
\\
\hline
Heap&
\raisebox{-18pt}{
$\underbrace{
\begin{tikzpicture}[every node/.style = {shape=rectangle,text width=6pt,text height=5pt,text centered}]
\mkarrow{0}{$s(t)$}
\mkcell{1}{}{$d$}
\mkarrow{1}{}
\mkdots{2}{}
\mksarrow{2}{}
\mkcell{3}{}{}
\mkarrow{3}{$s(u)$}
\end{tikzpicture}
}_{h_1}$
}
&
$+$
&
$h_2$
\end{tabular}
\end{center}

Fix values $a,b \in N$ such that 
$a \not\in
\Dom(h)\cup
\{s(p) \mid p\mapsto (\wild,\wild) \in \Vec{\psi}\} \cup
\{s(p),s(q) \mid \Ls(p,q) \in \Vec{\psi}\} \cup
\{s(p),s(q),s(p'),s(q') \mid \Dll(p,q,p',q') \hbox{ is in } \Vec{\psi}\} \cup
\bigcup \{[s(p),s(q)] \mid \Arr(p,q) \hbox{ is in } \Vec{\psi}\}$, and
$b \not\in \Dom(h)\cup\{a\} $. 

Define $h'_1$ by
$\Dom(h'_1) = \{s(t),a\}$, 
$h'_1(m) = (a,d)$ if $m = s(t)$, and
$h'_1(m) = (s(u),b)$ if $m \neq s(t)$. 
Let $s' = s[z:=a,y_1:=d,y_2:=b]$. Then we have
\[
s',h'_1+h_2 \models t\Pto (z,y_1) * z\Pto (u,y_2) * \phi, 
\]
which is depicted as follows: 
\begin{center}
\begin{tabular}{lccc}
Formula&$t \mapsto (z,y_1)*z\mapsto(u,y_2)$&$*$&$\phi$
\\
\hline
Heap&
\raisebox{-18pt}{
$\underbrace{
\begin{tikzpicture}[every node/.style = {shape=rectangle,text width=6pt,text height=5pt,text centered}]
\mkarrow{0}{$s(t)$}
\mkcell{1}{}{$d$}
\mkarrow{1}{$a$}
\mkcell{2}{}{$b$}
\mkarrow{2}{$s(u)$}
\end{tikzpicture}
}_{h'_1}$
}
&
$+$
&
$h_2$
\end{tabular}
\end{center}

Hence, by the assumption, we have 
\[
s',h'_1+h_2 \models \psi_j
\]
for some $\psi_j \in \Vec{\psi}$. 
Recall that $z,y_1,y_2$ do not appear in $\psi_j$ since they are fresh variables. 
So we have 
\[
s,h'_1+h_2 \models \psi_j. 
\]

Let $\psi_j \equiv \Pi\land\Sigma$. 
Then $s\models\Pi$, and there are $h_3$, $h_4$ and
an atomic spatial formula $\sigma$ such that
$s,h_3 \models \sigma$, $s,h_4 \models \Sigma-\sigma$, $a\in\Dom(h_3)$, and $h_3+h_4 = h'_1+h_2$, 
where $\Sigma-\sigma$ is the spatial formula obtained by removing $\sigma$ from $\Sigma$. 

We can show that $\sigma$ must have the form $\Ls(p,q)$ as follows:
it cannot be $p\Pto (\wild,\wild)$ by $a \neq s(p)$; 
it cannot be $\Arr(p,q)$ by $a \not\in [s(p),s(q)]$;
it cannot be $\Dll(p,q,p',q')$, otherwise, 
by $a \neq s(p),s(p'),b$ and $h_3(a) = (s(u),b)$,
we have $b \in \Dom(h_3)\setminus \{a\} \subseteq \Dom(h_3+h_4)\setminus \{a\} = \Dom(h'_1+h_2)\setminus \{a\} \subseteq \Dom(h)$, which contradicts the condition of $b$. 

Note that $a\neq s(p), s(q)$ by the condition of $a$. 
Hence we have
$s(t), a \in \Dom(h_3)$,
$h_3(s(t)) = h'_1(s(t)) = (a,d)$
and $h_3(a) = h'_1(a) = (s(u),b)$. 
Therefore the current situation of $h_3$ and $h_4$ is depicted as follows: 
\begin{center}
\begin{tabular}{lccc}
Formula&$\Ls(p,q)$&$*$&$\Sigma - \sigma$
\\
\hline
Heap&
\raisebox{-18pt}{
$\underbrace{
\begin{tikzpicture}[every node/.style = {shape=rectangle,text width=6pt,text height=5pt,text centered}]
\mkarrow{0}{$s(p)$}
\mkcell{1}{}{}
\mkarrow{1}{}
\mkdots{2}
\mksarrow{2}{$s(t)$}
\mkcell{3}{}{$d$}
\mkarrow{3}{$a$}
\mkcell{4}{}{$b$}
\mkarrow{4}{$s(u)$}
\mkdots{5}
\mksarrow{5}{$s(q)$}
\end{tikzpicture}
}_{h_3}$
}
&
$+$
&
$h_4$
\end{tabular}
\end{center}

Then define $h'_3$ by $\Dom(h'_3) = \Dom(h_1) + (\Dom(h_3)\setminus\{a,s(t)\})$, 
where the symbol $+$ is the disjoint union symbol, 
$h'_3(m) = h_1(m)$ if $m \in \Dom(h_1)$, and
$h'_3(m) = h_3(m)$ if $m \not\in \Dom(h_1)$. 
Then we have the following situation: 
\begin{center}
\begin{tabular}{lccc}
Formula&$\Ls(p,q)$&$*$&$\Sigma - \sigma$
\\
\hline
Heap&
\raisebox{-18pt}{
$\underbrace{
\begin{tikzpicture}[every node/.style = {shape=rectangle,text width=6pt,text height=5pt,text centered}]
\mkarrow{0}{$s(p)$}
\mkcell{1}{}{}
\mkarrow{1}{}
\mkdots{2}
\mksarrow{2}{$s(t)$}
\mkcell{3}{}{$d$}
\mkarrow{3}{}
\mkdots{4}
\mksarrow{4}{}
\mkcell{5}{}{}
\mkarrow{5}{$s(u)$}
\mkdots{6}
\mksarrow{6}{$s(q)$}
\end{tikzpicture}
}_{h'_3}$
}
&
$+$
&
$h_4$
\end{tabular}
\end{center}
Note that the above picture covers both cases of $n = 2$ and $n\ge 4$. 
It satisfies $s,h'_3\models\Ls(p,q)$. 
We have $h'_3+h_4 = h_1+h_2 = h$ 
by removing $h'_1$ from both sides of $h_3+h_4 = h'_1+h_2$ and adding $h_1$ to them. 
Hence we have $s,h'_3+h_4\models\Sigma$. 
Therefore we have (1) since $s,h\models\Vec{\psi}$ can be obtained.
\end{proof}

The unroll collapse property for the doubly-linked list predicate is given in the next proposition.
The key idea of its proof is similar to that of Proposition~\ref{prop:unroll-ls}: 
First the list $\Dll(t,u,v,w)$ is replaced by some doubly-linked list
$\Dll(t,u,v,w)$ of length 3 in the antecedent of (1), in order to show (1).
Then (2) applies to obtain the succedent. Then we can show that the
doubly-linked list of length 3 is contained in some list either $\Ls(p,q)$ or $\Dll(p,q,p',q')$
of the succedent.  By replacing the doubly-linked list of length 3 by $\Dll(t,u,v,w)$ of arbitrary
length, we obtain the succedent of (1), since the difference by the
replacement is absorbed by $\Ls(p,q)$ or $\Dll(p,q,p',q')$.


\begin{prop}[Unroll Collapse for $\Dll$]\label{prop:unroll-dll}
The following clauses are equivalent:

(1)
$\Dll(t,u,v,w) * \phi \models \Vec{\psi}$.

(2)
$t = u \land v=w \land \phi \models \Vec{\psi}$, \quad
$t = v \land t \Pto (u,w) * \phi \models \Vec{\psi}$, and
\\
\phantom{\quad\qquad}
$t \Pto (z,w) * z \Pto (v,t) * v \Pto (u,z) * \phi \models \Vec{\psi}$,
where $z$ is fresh.
\end{prop}
\begin{proof}[Proof of Proposition~\ref{prop:unroll-dll}]
From (1) to (2) is trivial.
We consider the inverse direction. 
Assume $s,h\models \Dll(t,u,v,w) * \phi$.
We will show $s,h\models \Vec{\psi}$.

By the assumption, there are $h_1$ and $h_2$ such that 
$s,h_1\models \Dll(t,u,v,w)$, 
$s,h_2\models \phi$, and 
$h = h_1+h_2$. 
Hence $s,h_1\models \Dll^n(t,u,v,w)$ for some $n$ (take the smallest one). 

We will show $s,h\models \Vec{\psi}$. 
The case of $n = 0$ does not happen since $\Dll^0(t,u,v,w)$ is unsatisfiable. 
In the case of $n = 1$, $n = 2$, or $n = 4$, we have $s,h\models \Vec{\psi}$ by (2). 

We consider the other cases. 
The current situation of $h_1$ and $h_2$ is depicted as follows: 
\begin{center}
\begin{tabular}{lccc}
Formula&$\Dll(t,u,v,w)$&$*$&$\phi$
\\
\hline
Heap&
\raisebox{-18pt}{
$\underbrace{
\begin{tikzpicture}[every node/.style = {shape=rectangle,text width=6pt,text height=5pt,text centered}]
\mksarrow{0}{$s(t)$}
\mkarrowr{0}{$s(w)$}
\mkcell{1}{}{}
\mkarrow{1}{}
\mksarrowr{1}{}
\mkdots{2}{}
\mksarrow{2}{}
\mkarrowr{2}{}
\mkcell{3}{}{}
\mkarrow{3}{$s(u)$}
\mksarrowr{3}{$s(v)$}
\end{tikzpicture}
}_{h_1}$
}
&
$+$
&
$h_2$
\end{tabular}
\end{center}

Fix a value $a \in N$ such that 
$a \not\in
\Dom(h)\cup
\{s(p),s(q),s(r) \mid p\mapsto (q,r) \hbox{ is in } \Vec{\psi}\} \cup
\{s(p),s(q) \mid \Ls(p,q) \hbox{ is in } \Vec{\psi}\} \cup
\{s(p),s(q),s(p'),s(q') \mid \Dll(p,q,p',q') \hbox{ is in } \Vec{\psi}\} \cup
\bigcup \{[s(p),s(q)] \mid \Arr(p,q) \hbox{ is in } \Vec{\psi}\}$. 

Define $h'_1$ by $\Dom(h'_1) = \{s(t),a,s(v)\}$, 
$h'_1(m) = (a,s(w))$ if $m = s(t)$, 
$h'_1(m) = (s(v),s(t))$ if $m = a$, and 
$h'_1(m) = (s(u),a)$ if $m = s(v)$. 

Then we have
\[
s[z:=a],h'_1+h_2 \models t\Pto (z,w) * z\Pto (v,t) * v\Pto(u,z) * \phi. 
\]
The current situation is as follows: 
\begin{center}
\begin{tabular}{lccc}
Formula&$t\Pto (z,w) * z\Pto (v,t) * v\Pto(u,z)$&$*$&$\phi$
\\
\hline
Heap&
\raisebox{-18pt}{
$\underbrace{
\begin{tikzpicture}[every node/.style = {shape=rectangle,text width=6pt,text height=5pt,text centered}]
\mksarrow{0}{$s(t)$}
\mkarrowr{0}{$s(w)$}
\mkcell{1}{}{}
\mkarrow{1}{$a$}
\mkarrowr{1}{$s(t)$}
\mkcell{2}{}{}
\mkarrow{2}{$s(v)$}
\mkarrowr{2}{$a$}
\mkcell{3}{}{}
\mkarrow{3}{$s(u)$}
\mksarrowr{3}{$s(v)$}
\end{tikzpicture}
}_{h'_1}$
}
&
$+$
&
$h_2$
\end{tabular}
\end{center}

Then, by the assumption, we have 
\[
s,h'_1+h_2 \models \psi_j
\]
for some $\psi_j \in \Vec{\psi}$, 
since $z$ does not appear in $\psi_j$. 

Let $\psi_j \equiv \Pi\land\Sigma$. 
Then $s\models\Pi$, and there are $h_3$, $h_4$ and $\sigma$ such that
$s,h_3 \models \sigma$,
$s,h_4 \models \Sigma-\sigma$,
$a\in\Dom(h_3)$, and
$h_3+h_4 = h'_1+h_2$.

We can show $\sigma$ have the form $\Ls(p,q)$ or $\Dll(p,q,p',q')$ as follows:
it cannot be $p\Pto (\wild,\wild)$ by $a \neq s(p)$; 
it cannot be $\Arr(p,q)$ by $a \not\in [s(p),s(q)]$. 

{\bfseries Case 1:} the case of $\sigma \equiv \Dll(p,q,p',q')$.
Note that $s(v),s(t),a \in \Dom(h_3)$,
since $a$ cannot be the first or last cell of the dll by $a \neq s(p),s(p')$.
Hence
$h_3(s(t)) = h'_1(s(t)) = (a,s(w))$
and
$h_3(a) = h'_1(a) = (s(v),s(t))$. 
This case is depicted as follows: 
\begin{center}
\begin{tabular}{lccc}
Formula&$\Dll(p,q,p',q')$&$*$&$\Sigma - \sigma$
\\
\hline
Heap&
\raisebox{-18pt}{
$\underbrace{
\begin{tikzpicture}[every node/.style = {shape=rectangle,text width=6pt,text height=5pt,text centered}]
\mksarrow{0}{$s(p)$}
\mkarrowr{0}{$s(q')$}
\mkcell{1}{}{}
\mkarrow{1}{}
\mksarrowr{1}{$s(p)$}
\mkdots{2}
\mksarrow{2}{$s(t)$}
\mkarrowr{2}{$s(w)$}
\mkcell{3}{}{}
\mkarrow{3}{$a$}
\mkarrowr{3}{$s(t)$}
\mkcell{4}{}{}
\mkarrow{4}{$s(v)$}
\mkarrowr{4}{$a$}
\mkcell{5}{}{}
\mkarrow{5}{$s(u)$}
\mksarrowr{5}{$s(v)$}
\mkdots{6}{}{}
\mksarrow{6}{}
\mkarrowr{6}{}
\mkcell{7}{}{}
\mkarrow{7}{$s(q)$}
\mksarrowr{7}{$s(p')$}
\end{tikzpicture}
}_{h_3}$
}
&
$+$
&
$h_4$
\end{tabular}
\end{center}

Define $h'_3$ by $\Dom(h'_3) = \Dom(h_1) + (\Dom(h_3)\setminus\{a,s(t),s(v)\})$, 
$h'_3(m) = h_1(m)$ if $m \in \Dom(h_1)$, and 
$h'_3(m) = h_3(m)$ if $m \not\in \Dom(h_1)$. 
Then we have the situation depicted as follows: 
\begin{center}
\begin{tabular}{lccc}
Formula&$\Dll(p,q,p',q')$&$*$&$\Sigma - \sigma$
\\
\hline
Heap&
\raisebox{-18pt}{
$\underbrace{
\begin{tikzpicture}[every node/.style = {shape=rectangle,text width=6pt,text height=5pt,text centered}]
\mksarrow{0}{$s(p)$}
\mkarrowr{0}{$s(q')$}
\mkcell{1}{}{}
\mkarrow{1}{}
\mksarrowr{1}{$s(p)$}
\mkdots{2}
\mksarrow{2}{$s(t)$}
\mkarrowr{2}{$s(w)$}
\mkcell{3}{}{}
\mkarrow{3}{}
\mksarrowr{3}{}
\mkdots{4}
\mksarrow{4}{}
\mkarrowr{4}{}
\mkcell{5}{}{}
\mkarrow{5}{$s(u)$}
\mksarrowr{5}{$s(v)$}
\mkdots{6}{}{}
\mksarrow{6}{}
\mkarrowr{6}{}
\mkcell{7}{}{}
\mkarrow{7}{$s(q)$}
\mksarrowr{7}{$s(p')$}
\end{tikzpicture}
}_{h'_3}$
}
&
$+$
&
$h_4$
\end{tabular}
\end{center}
Note that the above picture covers both cases of $n = 3$ and $n\ge 5$. 
It satisfies $s,h'_3\models\Dll(p,q,p',q')$.
We have $h'_3+h_4 = h_1+h_2 = h$ 
by removing $h'_1$ from both sides of $h_3+h_4 = h'_1+h_2$ and adding $h_1$ to them. 
Hence we have $s,h'_3+h_4\models\Sigma$. 
Therefore we have (1) since $s,h\models\Vec{\psi}$ can be obtained.

{\bfseries Case 2:} the case of $\sigma \equiv \Ls(p,q)$.
Note that $a,s(t) \in \Dom(h_3)$, 
since $a$ cannot be the first cell of the list by $a \neq s(p)$.
Hence
$h_3(s(t)) = h'_1(s(t)) = (a,s(w))$
and
$h_3(a) = h'_1(a) = (s(v),s(t))$. 
We also note that $s(v) \in \Dom(h_3)$ or $s(v) \in \Dom(h_4)$. 
The latter case implies $s(v) = s(q)$. 
We consider two subcases about where $s(v)$ is. 

{\bfseries Case 2.1:} the case of $s(v) \in \Dom(h_3)$. 
Consider $h'_3$ taken in the case 1. 
This case is depicted as follows: 
\begin{center}
\begin{tabular}{lccc}
Formula&$\Ls(p,q)$&$*$&$\Sigma - \sigma$
\\
\hline
Heap&
\raisebox{-18pt}{
$\underbrace{
\begin{tikzpicture}[every node/.style = {shape=rectangle,text width=6pt,text height=5pt,text centered}]
\mksarrow{0}{$s(p)$}
\mkcell{1}{}{}
\mkarrow{1}{}
\mkdots{2}
\mksarrow{2}{$s(t)$}
\mkcell{3}{}{}
\mkarrow{3}{}
\mksarrowr{3}{}
\mkdots{4}
\mksarrow{4}{}
\mkarrowr{4}{}
\mkcell{5}{}{}
\mkarrow{5}{$s(u)$}
\mkdots{6}{}{}
\mksarrow{6}{}
\mkcell{7}{}{}
\mkarrow{7}{$s(q)$}
\end{tikzpicture}
}_{h'_3}$
}
&
$+$
&
$h_4$
\end{tabular}
\end{center}
Note that the above picture covers both cases of $n = 3$ and $n\ge 5$. 
It satisfies $s,h'_3\models\Ls(p,q)$.
We have $h'_3+h_4 = h_1+h_2 = h$ 
by removing $h'_1$ from both sides of $h_3+h_4 = h'_1+h_2$ and adding $h_1$ to them. 
Hence we have $s,h'_3+h_4\models\Sigma$. 
Therefore we have (1) since $s,h\models\Vec{\psi}$ can be obtained.

{\bfseries Case 2.2:} the case of $s(v) \in \Dom(h_4)$. 
Recall that this case implies $s(v) = s(q)$. 
This case is depicted as follows: 
\begin{center}
\begin{tabular}{lccc}
Formula&$\Ls(p,q)$&$*$&$\Sigma - \sigma$
\\
\hline
Heap&
\raisebox{-18pt}{
$\underbrace{
\begin{tikzpicture}[every node/.style = {shape=rectangle,text width=6pt,text height=5pt,text centered}]
\mksarrow{0}{$s(p)$}
\mkcell{1}{}{}
\mkarrow{1}{}
\mkdots{2}
\mksarrow{2}{$s(t)$}
\mkcell{3}{}{}
\mkarrow{3}{$a$}
\mkarrowr{3}{$s(t)$}
\mkcell{4}{}{}
\mklarrow{4}{$s(v)=s(q)$}
\mkarrowr{4}{$a$}
\end{tikzpicture}
}_{h_3}$
}
&
$+$
&
\raisebox{-12pt}{
$\underbrace{
\begin{tikzpicture}[every node/.style = {shape=rectangle,text width=6pt,text height=5pt,text centered}]
\mksarrow{0}{$s(v)$}
\mkarrowr{0}{$a$}
\mkcell{1}{}{}
\mkarrow{1}{$s(u)$}
\mkdots{2}
\end{tikzpicture}
}_{h_4}$
}
\end{tabular}
\end{center}

Let $h_1(s(v)) = (s(u),d)$. 
Then we define $\tilde{h}'_3$ by 
$\Dom(\tilde{h}'_3) = (\Dom(h_1)\setminus\{s(v)\})+(\Dom(h_3)\setminus\{a,s(t)\})$, 
$\tilde{h}'_3(m) = h_1(m)$ if $m \in \Dom(h_1)\setminus\{s(v)\}$, and 
$\tilde{h}'_3(m) = h_3(m)$ otherwise. 
We also define $\tilde{h}_4$ by 
$\Dom(\tilde{h}_4) = \Dom(h_4)$, 
$\tilde{h}_4(m) = h_4(m)$ if $m\neq s(v)$, and
$\tilde{h}_4(m) = (s(u),d)$ if $m = s(v)$. 
Then we have the following situation: 
\begin{center}
\begin{tabular}{lccc}
Formula&$\Ls(p,q)$&$*$&$\Sigma - \sigma$
\\
\hline
Heap&
\raisebox{-18pt}{
$\underbrace{
\begin{tikzpicture}[every node/.style = {shape=rectangle,text width=6pt,text height=5pt,text centered}]
\mksarrow{0}{$s(p)$}
\mkcell{1}{}{}
\mkarrow{1}{}
\mkdots{2}
\mksarrow{2}{$s(t)$}
\mkcell{3}{}{}
\mkarrow{3}{}
\mksarrowr{3}{}
\mkdots{4}
\mksarrow{4}{}
\mkarrowr{4}{}
\mkcell{5}{}{}
\mklarrow{5}{$s(v)=s(q)$}
\mkarrowr{5}{$d$}
\end{tikzpicture}
}_{\tilde{h}'_3}$
}
&
$+$
&
\raisebox{-12pt}{
$\underbrace{
\begin{tikzpicture}[every node/.style = {shape=rectangle,text width=6pt,text height=5pt,text centered}]
\mksarrow{0}{$s(v)$}
\mkarrowr{0}{$d$}
\mkcell{1}{}{}
\mkarrow{1}{$s(u)$}
\mkdots{2}
\end{tikzpicture}
}_{\tilde{h}_4}$
}
\end{tabular}
\end{center}
Note that the above picture covers both cases of $n = 3$ and $n\ge 5$. 
Then we have 
$\tilde{h}'_3+\tilde{h}_4 = h_1 + h_2 = h$, 
by removing $h'_1$ from both sides of $h_3+h_4 = h'_1+h_2$ and adding $h_1$ to them. 
Hence we have $s,\tilde{h}'_3+\tilde{h}_4\models\Sigma$. 
Therefore we have (1) since $s,h\models\Vec{\psi}$ can be obtained.
\end{proof}

\begin{rem}\rm
We note that our unroll collapse properties (Proposition~\ref{prop:unroll-ls} and \ref{prop:unroll-dll})
hold for entailments with the points-to predicate, the array predicate, 
and the (possibly cyclic) singly-linked and doubly-linked list predicates. 
We also note that ours hold for entailments that contain arithmetic. 
The original version of unroll collapse is given by Berdine et al.~\cite{OHearn04}. 
It holds for entailments with only the points-to predicate and the acyclic singly-linked list predicate. 
We cannot compare ours and theirs naively, since the singly-linked list predicates in both papers are different for cyclic lists.
\end{rem}
  





%% file: 7.tex
\section{Decision Procedure for Arrays and Lists}

This section gives our algorithm for checking 
the validity of a given entailment, whose antecedent do not contain list predicates. 
In the following subsections~\ref{subsect:proofsystem} and \ref{subsect:proofsearch}, we implicitly assume that the antecedent of each entailment is list-free. 

The procedure first eliminates the list predicates in the succedent of a given entailment. 
The resulting entailments only contain the points-to and array predicates, that is, they are entailments in $\SLAR$. 
Then the procedure checks their validity by using the decision procedure of $\SLAR$. 

\subsection{Proof System for Elimination of Lists in Succedents}\label{subsect:proofsystem}

This subsection gives a proof system for entailments 
whose antecedents do not contain list predicates. 
Our decision procedure is given as a proof-search procedure 
of the proof system. 

We define a Presburger formula $\Cell{\Sigma}{t}$
that means the address $t$ is a cell of $\Sigma$. 
We also define $\Term(\Sigma)$, which is the set of terms in $\Sigma$. 
For defining them, we will not implicitly use the commutative law for $*$ and the unit law for $\Emp$. 

\begin{defi}\rm
$\Cell{\Sigma}{t}$ is inductively defined as follows: 
\begin{center}
\begin{tabular}{l@{\qquad}l}
$\Cell{\Emp}{t} \eqDef \False$, 
&
$\Cell{\Emp*\Sigma'}{t} \eqDef \Cell{\Sigma'}{t}$, 
\\
$\Cell{t'\Pto(\wild,\wild)}{t} \eqDef t=t'$, 
&
$\Cell{t'\Pto(\wild,\wild)*\Sigma'}{t} \eqDef t=t'\vee\Cell{\Sigma'}{t}$, 
\\
$\Cell{\Arr(t',u')}{t} \eqDef t'\le t \le u'$, 
&
$\Cell{\Arr(t',u')*\Sigma}{t} \eqDef t'\le t \le u' \vee \Cell{\Sigma}{t}$. 
\end{tabular}
\end{center}
\end{defi}

\begin{defi}\rm
$\Term(\Sigma)$ is inductively defined as follows: 

\qquad
\begin{tabular}{l@{\qquad}l@{\qquad}l}
$\Term(\Emp) \eqDef \emptyset$, 
&
$\Term(t\Pto(u_1,u_2)) \eqDef \{t,u_1,u_2\}$, 
&
$\Term(\Arr(t,u)) \eqDef \{t,u\}$, 
\\
\multicolumn{3}{l}{
$\Term(\Sigma_1*\Sigma_2) \eqDef \Term(\Sigma_1) \cup \Term(\Sigma_2)$. 
}
\end{tabular}
\end{defi}
We write $\Term(\Vec\Sigma)$ for $\bigcup_{\Sigma\in\Vec\Sigma}\Term(\Sigma)$. 

\begin{lem}\label{lem:cell}\rm
Suppose $s,h\models\Sigma$. Then 
$s\models\Cell{\Sigma}{t}$ 
if and only if 
$s(t) \in \Dom(h)$. 
\end{lem}
\begin{proof}
The claim is shown by induction on $\Sigma$.
\end{proof}

Then we define the inference rules which give our algorithm. 
The rules are shown in Figure~\ref{fig:inferencerules}. 

\begin{figure}[t]
  \rule{\textwidth}{1pt}
  \\[10pt]
  \scalebox{0.8}{
\begin{tabular}{c}
$\infer[\hbox{(Start)}]{
\varphi \vdash \Vec{\psi}
        }{}$
where $\varphi$ is list-free and $\varphi \models \Vec{\psi}$ in $\SLAR$
\\[\SKIP]
$\infer[\hbox{(UnsatL)}]{
\varphi \vdash \Vec{\psi}
        }{}$
if $\varphi$ is unsat
\hspace{2cm}
$\infer[\hbox{(UnsatR)}]{
\varphi \vdash \psi,\Vec{\psi}
        }{
        \varphi \vdash \Vec{\psi}
        }$
if $\varphi\land\psi$ is unsat
\\[\SKIP]
$\infer[\hbox{($\Pto$LsEM)}]{
        t\Pto (v,w) * \varphi \vdash \Ls(t',u')*\psi,\Vec{\psi}
        }{
        \deduce{
        t=t'\land t\Pto (v,w) * \varphi \vdash \Ls(t',u')*\psi,\Vec{\psi}
        }{
        t\neq t'\land t\Pto (v,w) * \varphi \vdash \Ls(t',u')*\psi,\Vec{\psi}
        }
        }$
\\
\hfill
if $t=t'\land\Pi_\varphi$ 
and $t \neq t'\land\Pi_\varphi$ are satisfiable
\\[\SKIP]
$\infer[\hbox{($\Pto$Ls)}]{
        t\Pto (v,w) * \varphi \vdash \Ls(t',u')*\psi,\Vec{\psi}
        }{
        t\Pto (v,w) * \varphi \vdash t'=u'\land\psi,t'\Pto (v,w) * \Ls(v,u')*\psi,\Vec{\psi}
        }$
if $\Pi_\varphi\models t=t'$
\\[\SKIP]
$\infer[\hbox{(LsElim)}]{
        \varphi \vdash \Ls(t',u')*\psi,\Vec{\psi}
        }{
        \varphi \vdash t'=u'\land\psi,\Vec{\psi}
        }$
if $\Pi_\varphi\not\models \Cell{\Sigma_\varphi}{t'}$
\\[\SKIP]
$\infer[\hbox{($\Pto$DllEM)}]{
        t\Pto (v,w) * \varphi \vdash \Dll(t',u',v',w')*\psi,\Vec{\psi}
        }{
        \deduce{
        t=t'\land t\Pto (v,w) * \varphi \vdash \Dll(t',u',v',w')*\psi,\Vec{\psi}
        }{
        t\neq t'\land t\Pto (v,w) * \varphi \vdash \Dll(t',u',v',w')*\psi,\Vec{\psi}
        }
        }$
\\
\hfill
if $t=t'\land\Pi_\varphi$ 
and $t \neq t'\land\Pi_\varphi$ are satisfiable
\\[\SKIP]
$\infer[\hbox{($\Pto$Dll)}]{
        t\Pto (v,w) * \varphi 
                \vdash 
                \Dll(t',u',v',w')*\psi,
                \Vec{\psi}
        }{
        t\Pto (v,w) * \varphi 
                \vdash 
                t'=u'\land v'=w' \land \psi,
                t\Pto (v,w') * \Dll(v,u',v',t')*\psi,
                \Vec{\psi}
        }$
if $\Pi_\varphi\models t=t'$
\\[\SKIP]
$\infer[\hbox{(DllElim)}]{
        \varphi \vdash \Dll(t',u',v',w')*\psi,\Vec{\psi}
        }{
        \varphi \vdash t'=u'\land v'=w' \land \psi,\Vec{\psi}
        }$
if $\Pi_\varphi\not\models \Cell{\Sigma_\varphi}{t'}$
\\[\SKIP]
$\infer[\hbox{($\Arr$ListEM)}]{
        \Arr(t,v) * \varphi \vdash L(t',\Vec{u'})*\psi,\Vec{\psi}
        }{
        \deduce{
                t\le t'\le v \land \Arr(t,v) * \varphi 
                        \vdash 
                        L(t',\Vec{u'})*\psi,\Vec{\psi}
                }{
                \deduce{
                  v < t' \land \Arr(t,v) * \varphi
                        \vdash 
                        L(t',\Vec{u'})*\psi,\Vec{\psi}
                  }{
                  t' < t \land \Arr(t,v) * \varphi 
                        \vdash 
                        L(t',\Vec{u'})*\psi,\Vec{\psi}
                  }
                }
        }$,
\\
\hfill
where $t\le t'\le v \land \Pi_\varphi$ and $(t'< t \vee v < t') \land \Pi_\varphi$ are satisfiable
\\
\hfill
$L(t',\Vec{u'})$ is $\Ls(t',u')$ or $\Dll(t',u'_1,u'_2,u'_3)$
\\[\SKIP]
$\infer[\hbox{($\Arr$Ls)}]{
        \Arr(t,v) * \varphi \vdash \Ls(t',u')*\psi,\Vec{\psi}
        }{
        \Arr(t,v) * \varphi \vdash t'=u'\land\psi,\Vec{\psi}
        }$
\quad 
if $\Pi_\varphi \models t\le t'\le v$
\\[\SKIP]
$\infer[\hbox{($\Arr$Dll)}]{
        \Arr(t,v) * \varphi \vdash \Dll(t',u',v',w')*\psi,\Vec{\psi}
        }{
        \Arr(t,v) * \varphi \vdash t'=u'\land v'=w'\land\psi,\Vec{\psi}
        }$
\quad
if $\Pi_\varphi \models t\le t'\le v$
\\[\SKIP]
\end{tabular}
}
  \rule{\textwidth}{1pt}
  \\[10pt]
\caption{Inference rules for the decision procedure}
\label{fig:inferencerules}
\end{figure}







Let $h$ and $h'$ be heaps such that $\Dom(h) = \Dom(h')$. 
We write $h\sim_d h'$ if $h(x) = h'(x)$ holds for any $x \in \Dom(h)\setminus\{d\}$. 

The following lemma is used in the proof of 
local completeness of the inference rules. 

\begin{lem}\label{lem:freshreplace}
(1) Suppose that $s,h\models\sigma$, $(a,b) \in \Ran(h)$ and 
$a\not\in\Dom(h) \cup \{s(t) \mid t \in \Term(\sigma) \}$. 
Then $\sigma$ is an array atomic formula. 

(2) Suppose that $s,h\models\Sigma$, $h(d) = (a,b)$, 
$a\not\in\Dom(h) \cup \{s(t) \mid t \in \Term(\Sigma) \}$ 
and $h'\sim_d h$. 
Then $s,h'\models\Sigma$.
\end{lem}
\begin{proof}
(1) We show the claim by case analysis of $\sigma$. 

The case that $\sigma$ is $t \Pto (u,v)$. 
By the assumptions, we have $(a,b) \in \Ran(h) = \{(s(u),s(v))\}$. 
Hence we obtain $a = s(u) \in \{s(t) \mid t \in \Term(\sigma) \}$,
which contradicts the assumption. 

The case that $\sigma$ is $\Ls(t,u)$. 
By the assumptions, we have $h\neq \emptyset$. 
Hence $h$ contains a non-empty list that starts from $t$. 
Then $a$ must be a cell of the list, 
since $a \neq s(u)$. 
Therefore we have $a \in \Dom(h)$, which contradicts the assumption. 

The case of $\Dll(t,u,v,w)$ can be shown in a similar way to 
the case of $\Ls(t,u)$. 

(2) By the assumptions, there exist $h_1$ and $h_2$ such that 
$s,h_1\models\sigma$, $s,h_2\models\Sigma-\sigma$, 
$h = h_1+h_2$ and $d \in \Dom(h_1)$. 
By (1), $\sigma$ is an array atomic formula. 
Let $h'_1$ be $h'|_{\Dom(h_1)}$. 
Then we have $s,h'_1\models\sigma$, since $h_1\sim_d h'_1$. 
Thus we obtain $s,h'\models\Sigma$, since $h' = h'_1+h_2$. 
\end{proof}

\begin{prop}\label{prop:sound_complete}\rm
Each inference rule is sound and locally complete. 
\end{prop}
\begin{proof}
The claims of the rules 
(Start), (UnsatL) and (UnsatR) 
are immediately shown. 
The claims of the rules 
($\Pto$LsEM), ($\Pto$DllEM), (ArrListEM), ($\Pto$Ls) and ($\Pto$Dll)
are easily shown (without the side-conditions). 

Soundness of the rule (LsElim) is easily shown. 
We show local completeness of it. 
Assume $\varphi\models\Ls(t',u')*\psi,\Vec{\psi}$, 
the validity $s,h\models\varphi$ and $\Pi_\varphi\not\models\Cell{\Sigma_\varphi}{t'}$. 
We will show $s,h\models t'=u'\land\psi,\Vec{\psi}$. 
By the assumption, 
we have $s,h\models \Ls(t',u')*\psi$ or $s,h\models \Vec{\psi}$. 
If the latter case holds, then we have the claim. 
Otherwise there exist $h_1$ and $h_2$ such that 
$h = h_1+h_2$, $s,h_1\models\Ls(t',u')$ and $s,h_2\models\psi$. 
By the lemma~\ref{lem:cell}, we have $s(t')\not\in\Dom(h)$ 
since $s,h\models\Sigma_\varphi$ and $s\not\models\Cell{\Sigma_\varphi}{t'}$.  
Hence we obtain $h_1=\emptyset$, $s\models t'=u'$ and $s,h\models\psi$. 
Thus $s,h\models t'=u'\land\psi$ holds. 

Soundness of the rule (DllElim) is shown immediately. 
Local completeness of it can be shown in a similar way to 
the proof of local completeness of (LsElim). 

Soundness of the rule (ArrLs) is easily shown. 
We show local completeness of it. 
Assume $\Arr(t,v)*\varphi\models\Ls(t',u')*\psi,\Vec{\psi}$, 
the validity $s,h\models\Arr(t,v)*\varphi$ and $\Pi_\varphi\models t\le t'\le v$. 
We will show $s,h\models t'=u'\land\psi,\Vec{\psi}$. 
Let $(a,b)$ be $h(s(t'))$. 
Fix a fresh value $a'$ such that 
$a' \not\in \Dom(h)\cup\{s(t) \mid t \in\Term(\Ls(t',u')*\Sigma_\psi,\Sigma_{\Vec{\psi}})\}$. 
Define $h'$ by $h' \sim_{s(t')} h$ and $h'(s(t')) = (a',b)$. 
Then we have $s,h'\models\Arr(t,v)*\varphi$ since $s\models t\le t' \le v$ by the side condition. 
Hence $s,h'\models\Ls(t',u')*\psi,\Vec{\psi}$ holds. 
If $s,h'\models\Vec{\psi}$, then 
we have the claim
by Lemma~\ref{lem:freshreplace}~(2). 
Otherwise we have $s,h'\models\Ls(t',u')*\psi$. 
Hence there exist $h'_1$ and $h'_2$ such that 
$h' = h'_1+h'_2$, $s,h'_1\models\Ls(t',u')$ and $s,h'_2\models\psi$. 
Then we can show $h'_1 = \emptyset$ as follows: 
Suppose it does not hold, then 
$s(t')\in\Dom(h'_1)$ and $(a',b)\in\Ran(h'_1)$; 
hence we have a contradiction since 
$\Ls(t',u')$ is an array atomic formula by Lemma~\ref{lem:freshreplace}~(1). 
Therefore we have $s\models t'=u'$ and $s,h'\models\psi$. 
Finally we have $s,h\models t'=u' \land \psi$ by Lemma~\ref{lem:freshreplace}~(2). 

The claim of the rule (ArrDll) can be shown in a similar way to (ArrLs).
\end{proof}

\subsection{Proof Search Algorithm}\label{subsect:proofsearch}

In our decision procedure, 
we read each inference rule 
from the bottom (conclusion) to the top (assumptions). 
Let $\Entl$ be the set of entailments whose antecedents 
do not contain list predicates. 
For each inference rule $\mathcal{R}$ 
we define 
a partial function $\Apply{\R}$ from $\Entl$ to $\Pow(\Entl)$ as follows. 
$\Apply{\R}(J)$ is defined when $J$ is 
a conclusion of some instance of $\R$ (including its side condition). 
$\Apply{\R}(J)$ is the assumptions in some instance of $\R$ with the conclusion $J$ (non-deterministically chosen). 

Let $\tau$ be a derivation tree
\raisebox{-10pt}{
$\infer[\R]{J}{\tau_1&\ldots &\tau_k}$ 
}. 
We sometimes represent this tree by a tuple 
$\Deriv{\R}{J}{\tau'_1,\ldots,\tau'_k}$,
where $\tau'_1,\ldots,\tau'_k$ are representation of $\tau_1,\ldots,\tau_k$, respectively. 
Our proof search procedure $\Search$ is given in Figure~\ref{fig:search}.  
$\Search(J)$ returns a tuple for a derivation tree of $J$ or returns $\Fail$. 

\def\ln#1{\hbox{\tiny #1}}

\begin{figure}[t]
\begin{tabular}{l}
\FunctionK$\Search(J)$
\\
\qqOne\IfK{UnsatL is applicable to $J$} $\R\ :=\ {\rm UnsatL}$
\\
\qqTwo\ElifK{Start is applicable to $J$} $\R\ :=\ {\rm Start}$
\\
\qqThr\ElifK{UnsatR is applicable to $J$} $\R\ :=\ {\rm UnsatR}$
\\
\qqFor\ElifK{$\Pto$LsEM is applicable to $J$} $\R\ :=\ \Pto{\rm LsEM}$
\\
\qqFiv\ElifK{ArrListEM is applicable to $J$} $\R\ :=\ {\rm ArrListEM}$
\\
\qqSix\ElifK{LsElim is applicable to $J$} $\R\ :=\ {\rm LsElim}$
\\
\qqSev\ElifK{ArrLs is applicable to $J$} $\R\ :=\ {\rm ArrLs}$
\\
\qqEig\ElifK{$\Pto$Ls is applicable to $J$} $\R\ :=\ \Pto{\rm Ls}$
\\
\qqNin\ElifK{$\Pto$DllEm is applicable to $J$} $\R\ :=\ \Pto{\rm DllEM}$
\\
\qqTen\ElifK{DllElim is applicable to $J$} $\R\ :=\ {\rm DllElim}$
\\
\qqEle\ElifK{ArrDll is applicable to $J$} $\R:=\ {\rm ArrDll}$
\\
\qqTwe\ElifK{$\Pto$Dll is applicable to $J$} $\R\ :=\ \Pto{\rm Dll}$
\\
\qqTht\ElseK \ReturnK{\Fail}
\\
\qqOne$\{J_1,\ldots,J_k\}\ :=\ \Apply{\R}(J)$
\\
\qqOne\IfK{each of $\Search(J_i)$ returns a tuple}
\ReturnK{$\Deriv{\R}{J}{\Search(J_1),\ldots,\Search(J_k)}$}
\\
\qqOne\ElseK \ReturnK{$\Fail$}
\end{tabular}
\caption{Proof search algorithm}
\label{fig:search}
\end{figure}

We first show the termination property of $\Search$. 
In order to show this, we define some notations. 
\begin{defi}\rm
(1) $\numPto(\psi)$ is the number of $\Pto$ in $\psi$, 
$\numList(\psi)$ is the number of $\Ls$ and $\Dll$ in $\psi$.

(2) $\degUnfold{\psi}{\varphi}$ is defined by 
$\numPto(\varphi)-\numPto(\psi)$. 

(3) Let $L(t',\Vec{u'})$ be $\Ls(t',\Vec{u'})$ or $\Dll(t',\Vec{u'})$, 
where $\Vec{u'}$ has an appropriate length. 
$\Deg(\Pi,\sigma,\sigma')$ is defined as follows:

$\Deg(\Pi,t\Pto(\wild,\wild),L(t',\wild)) = 
\left\{
\begin{array}{l@{\quad}l}
1
&
\hbox{if $t=t'\land\Pi$ and $t\neq t'\land\Pi$ are satisfiable,}
\\
0
&
\hbox{otherwise.}
\end{array}
\right.$

$\Deg(\Pi,\Arr(t,u),L(t',\wild)) = 
\left\{\!\!
\begin{array}{l@{~~}l}
1
&
\hbox{if $t\le t' \le u \land \Pi$ and $(t'<t \vee u < t')\land\Pi$ are satisfiable,}
\\
0
&
\hbox{otherwise.}
\end{array}
\right.$

$\Deg(\Pi,\sigma,\sigma') = 0$ \qquad if $\sigma'$ is not a list predicate. 

Then we define 
$\degEM{\psi}{\varphi}$ 
by 
$\underset{\sigma \hbox{ in } \varphi,\sigma' \hbox{ in } \psi}{\sum}\Deg(\Pi_\varphi,\sigma,\sigma')$. 

(4) 
We define the degree $\Rank{\psi}{\varphi}$ 
of $\psi$ with respect to $\varphi$ by 
$(\numList(\psi),\degUnfold{\psi}{\varphi},\degEM{\psi}{\varphi})$. 
The order on degrees is given by the lexicographic order. 

(5)
Let $J$ be an entailment $\varphi\vdash \{\psi_i\}_{i\in I}$. 
We define the measure $\widetilde{J}$ of $J$ as 
the sequence of $\Rank{\psi_i}{\varphi}$ $(i\in I)$ 
sorted in decreasing order. 
Then we write $J_1 < J_2$ for
$\widetilde{J_1} <_{{\rm lex}} \widetilde{J_2}$. 
\end{defi}

$\degUnfold{\psi}{\varphi}$ gives an upper bound 
of unfolding of list predicates 
in $\psi$ under $\varphi$. 
During the proof search, 
a list predicate in $\psi$ is unfolded 
if a matched points-to atomic formula in $\varphi$ is found. 
So, if unfolding is done more than the upper bound, 
$\psi$ becomes unsatisfiable 
since some points-to atomic formula has to match more than once. 

Note that the relation $<$ on entailments is 
a well-founded preorder, that is, 
there is no infinite decreasing chain, 
since the lexicographic order $\le_{{\rm lex}}$ on measures is a well-order. 

\begin{lem}\label{lem:order}
The degrees and the preorder on entailments satisfy the following properties. 

(1) 
$\Vec{\psi_1} \subsetneq \Vec{\psi_2}$ 
implies 
$(\varphi\vdash\Vec{\psi_1}) < (\varphi\vdash\Vec{\psi_2})$. 

(2) 
$\Rank{\psi_1}{\varphi},\ldots,\Rank{\psi_k}{\varphi} < \Rank{\psi}{\varphi}$ 
implies 
$(\varphi\vdash\Vec{\psi},\psi_1,\ldots,\psi_k) < (\varphi\vdash\Vec{\psi},\psi)$. 

(3) 
If 
$\Pi\subseteq\Pi'$, 
then 
$\Rank{\psi}{\Pi'\land\varphi} \le \Rank{\psi}{\Pi\land\varphi}$. 

(4) 
Suppose that 
$t = t'\land \Pi_\varphi$ 
and 
$t \neq t' \land \Pi_\varphi$ 
are satisfiable, 
$\varphi = t\Pto(\wild,\wild) * \varphi'$ 
and 
$\psi = L(t',\wild) * \psi'$. 
Then 
$(t = u \land \varphi\vdash\Vec{\psi},\psi), 
(t \neq u \land \varphi\vdash\Vec{\psi},\psi)
< 
(\varphi\vdash\Vec{\psi},\psi)$. 

(5) 
Suppose that 
$t \le t' \le u \land \Pi_\varphi$ 
and 
$(t' < t \vee u < t') \land \Pi_\varphi$ 
are satisfiable, 
$\varphi = \Arr(t,u)*\varphi'$ 
and 
$\psi = L(t',\wild) * \psi'$. 
Then 
$(t \le t' \le u \land \varphi\vdash\Vec{\psi},\psi), 
(t' < t \land \varphi\vdash\Vec{\psi},\psi),
(u < t' \land \varphi\vdash\Vec{\psi},\psi)
< 
(\varphi\vdash\Vec{\psi},\psi)$. 
\end{lem}
\begin{proof}
(1) and (2) are immediately shown since 
the order $\le_{lex}$ on measures is a lexicographic order and 
the preorder $<$ on entailments is defined by $<_{lex}$. 

(3) Assume $\Pi\subseteq \Pi'$. 
Take arbitrary $\varphi$ and $\psi$. 
We have 
$\Deg(\Pi'\land\Pi_\varphi,\sigma,\sigma') \le \Deg(\Pi\land\Pi_\varphi,\sigma,\sigma')$ 
since 
$\Pi\land\Pi_\varphi$ is satisfiable if $\Pi'\land\Pi_\varphi$ is satisfiable. 
Hence $\degEM{\psi}{\Pi'\land\varphi} \le \degEM{\psi}{\Pi\land\varphi}$ holds. 
Therefore we have 
$\Rank{\psi}{\Pi'\land\varphi} \le \Rank{\psi}{\Pi\land\varphi}$. 

(4) Suppose the assumption of the claim (4). 
We have 
$\degEM{\psi}{t=t'\land\varphi} < \degEM{\psi}{\varphi}$ 
since 
$\Deg(t=t'\land\Pi_\varphi,t\Pto(\wild,\wild),L(t',\wild)) < \Deg(\Pi_\varphi,t\Pto(\wild,\wild),L(t',\wild))$. 
Hence $\Rank{\psi}{t=t'\land\varphi} < \Rank{\psi}{\varphi}$ holds. 
We also have 
$\Rank{\psi_1}{t=t'\land\varphi} < \Rank{\psi_1}{\varphi}$ 
by (3). 
Thus we obtain 
$(t=t'\land\varphi\vdash\Vec{\psi},\psi) < (\varphi\vdash\Vec{\psi},\psi)$. 
Similarly we can also show $(t\neq t'\land\varphi\vdash\Vec{\psi},\psi) < (\varphi\vdash\Vec{\psi},\psi)$. 

(5) is shown in a similar way to (4). 
\end{proof}

\begin{lem}\label{lem:key-for-termination}
Let $\R$ be an inference rule 
other than (Start) and (UnsatL). 
Then we have $J' < J$ for any $J' \in \Apply{\R}(J)$. 
\end{lem}
\begin{proof}
We show the claim for each rule of $\R$. 

The claim for (UnsatR) is shown by Lemma~\ref{lem:order}~(1). 

The claims for ($\Pto$LsEM), ($\Pto$DllEM) and (ArrListEM) are shown by Lemma~\ref{lem:order}~(4). 

The claims for (LsElim), (DllElim), (ArrLs) and (ArrDll) are shown by Lemma~\ref{lem:order}~(2) since the number of $\Ls$ and $\Dll$ is reduced. 

The claim for ($\Pto$Ls) is shown as follows: 
Let $J$ and $J'$ be 
$t\Pto(v,w)*\varphi\vdash\Ls(t',u')*\psi,\Vec{\psi}$ 
and 
$t\Pto(v,w)*\varphi\vdash t'=u'\land\psi,t\Pto(v,w)*\Ls(v,u')*\psi,\Vec{\psi}$, 
respectively. 
Then we have 
$\Rank{t'=u'\land\psi}{t\Pto(v,w)*\varphi} < \Rank{\Ls(t',u')*\psi}{t\Pto(v,w)*\varphi}$ 
since the number of $\Ls$ is reduced. 
We also have 
$\Rank{t\Pto(v,w)*\Ls(v,u')*\psi}{t\Pto(v,w)*\varphi} < \Rank{\Ls(t',u')*\psi}{t\Pto(v,w)*\varphi}$ 
since 
$\numList{(t\Pto(v,w)*\Ls(v,u')*\psi)} = \numList{(\Ls(t',u')*\psi)}$ 
and 
$\degUnfold{t\Pto(v,w)*\Ls(v,u')*\psi}{t\Pto(v,w)*\varphi} < \degUnfold{\Ls(t',u')*\psi}{t\Pto(v,w)*\varphi}$. 
Hence we have $J' < J$ by Lemma~\ref{lem:order} (2). 

The claim for ($\Pto$Dll) is shown similarly to the case ($\Pto$Ls). 
\end{proof}

By the above lemma, we can show termination of $\Search$. 

\begin{lem}[Termination]\label{lem:termination}
$\Search(J)$ terminates for any $J$. 
\end{lem}
\begin{proof}
Suppose that $\Search$ does not terminate with an input $J_0$. 
Then there is an infinite sequence of recursive calls 
$\Search(J_0), \Search(J_1), \Search(J_2),\ldots$. 
By Lemma~\ref{lem:key-for-termination}, 
we have an infinite decreasing sequence $J_0 > J_1 > J_2 > \ldots$. 
This contradicts the well-foundedness of $<$. 
\end{proof}

\begin{lem}\label{lem:progress}
Let $J$ be a valid entailment. 
Then either of 
(UnsatL), 
(Start), 
(UnsatR), 
($\Pto$LsEM), 
(ArrListEM), 
(LsElim),
(ArrLs), 
($\Pto$Ls),
($\Pto$DllEM), 
(DllElim),
(ArrDll), or 
($\Pto$Dll)
is applicable to $J$. 
\end{lem}
\begin{proof}
If the antecedent of $J$ is unsatisfiable, then (UnsatL) is applicable to $J$.
If $J$ is list-free, then (Start) is applicable, since $J$ is valid. 
We consider the other cases. 

Let $(\varphi\vdash\Vec{\psi}) = J$. 
If there is $\varphi \in \Vec{\psi}$ 
such that $\varphi\land\psi$ is unsatisfiable, 
then (UnsatR) is applicable. 
In the following, we assume that $\varphi\land\psi$ is satisfiable 
for any $\psi\in\Vec{\psi}$. 
If $\varphi$ does not contain either $\Pto$ or $\Arr$, 
then (LsElim) or (DllElim) is applicable to $J$. 
Otherwise $\varphi$ contains $\Pto$ or $\Arr$, 
and $\Vec{\psi}$ contains list predicates. 

(1) 
The case that $\Vec{\psi}$ contains the $\Ls$ predicate. 
Fix $\Ls(t',u')$ in $\Vec{\psi}$. 
We consider the following subcases. 

(1.1) 
The case that there is $t\Pto(\wild,\wild)$ in $\varphi$ 
such that $t=t'\land\Pi_\varphi$ and $t\neq t'\land\Pi_\varphi$ 
are satisfiable. Then ($\Pto$EM) is applicable. 

(1.2) 
The case that there is $t\Pto(\wild,\wild)$ in $\varphi$ 
such that $t\neq t'\land\Pi_\varphi$ is unsatisfiable
(that is $\Pi_\varphi\models t=t'$). 
Then ($\Pto$Ls) is applicable. 

(1.3) 
The case that 
$t = t'\land\Pi_\varphi$ is unsatisfiable 
(that is $\Pi_\varphi\models t\neq t'$) 
for all $t\Pto(\wild,\wild)$ in $\varphi$. 
We consider the following subsubcases. 

(1.3.1) 
The case that there is $\Arr(t,u)$ in $\varphi$ 
such that 
$t\le t'\le u\land\Pi_\varphi$ 
and 
$(t'<t \vee u<t')\land\varphi_\varphi$ 
are satisfiable. 
Then (ArrLsEM) is applicable. 

(1.3.2) 
The case that there is $\Arr(t,u)$ in $\varphi$ 
such that 
$(t'<t \vee u<t')\land\Pi_\varphi$ 
is unsatisfiable
(that is $\Pi_\varphi\models t\le t' \le u$). 
Then (ArrLs) is applicable. 

(1.3.3)
The case that 
$t\le t'\le u\land\Pi_\varphi$ 
is unsatisfiable 
(that is $\Pi_\varphi\models t'<t \vee u < t'$) 
for all $\Arr(t,u)$ in $\varphi$. 
This case (LsElim) is applicable since $\Pi_\varphi\not\models\Cell{\Sigma_\varphi}{t'}$. 

(2)
The case that $\Vec{\psi}$ does not contain the $\Ls$ predicate. 
There exists $\Dll(t',u',v',w')$ in $\Vec{\psi}$. 
We can show the claim for this case in a similar way to (1). 
\end{proof}

By using the results of this section, we can show correctness of the proof search algorithm. 

\begin{prop}[Correctness]\label{prop:correctness}
$J$ is valid if and only if $\Search(J)$ returns a proof of $J$. 
\end{prop}
\begin{proof}
The if-part is shown by soundness of the proof system (Proposition~\ref{prop:sound_complete}). 
We show the only-if part. 
Assume that $J$ is a valid entailment. 
By Lemma~\ref{lem:termination}, $\Search(J)$ terminates. 
We show the claim by induction on computation of $\Search(J)$. 
By Lemma~\ref{lem:progress}, some rule $\R$ is applicable to $J$. 
Then $J_1,\ldots,J_k$ are obtained by $\Apply{\R}(J)$. 
By local completeness of the proof system (Proposition~\ref{prop:sound_complete}), 
each $J_k$ is valid. 
Hence, by the induction hypothesis, 
$\Search(J_i)$ returns a tuple that represents a proof of $J_i$
for any $i\in\{1,\ldots,k\}$. 
Therefore $\Search(J)$ returns a tuple that represents a proof of $J$. 
\end{proof}

%% file: 8.tex
\section{Decidability of Entailment Problem for $\SLLA$}

This section shows the second theorem of this paper, 
namely the decidability of the entailment problem of $\SLLA$, 
by combining the four results of the previous sections. 

\begin{thm}[Decidability of $\SLLA$]\label{thm:decidability-slla}
Checking the validity of entailments in $\SLLA$ is decidable. 
\end{thm}
\begin{proof}
We first give the decision procedure for $\SLLA$. 
An entailment $J$ of $\SLLA$ is given as an input for the procedure. 
Then it performs as follows. 
(i) The decision procedure eliminates the list predicates that appear in the antecedent of $J$ 
by using the unroll collapse (Propositions~\ref{prop:unroll-ls} and \ref{prop:unroll-dll}). 
Then it obtains entailments $J_1,\ldots,J_k$ whose antecedents are list-free. 
(ii) It computes $\Search(J_i)$ ($i=1,\ldots,k$). 
It returns ``valid'' if each of $\Search(J_i)$ returns a tuple. 
Otherwise it returns ``invalid''. 

Termination property of the decision procedure can be obtained from
the termination property of $\Search$ (Proposition~\ref{lem:termination}). 
Correctness of the decision procedure is stated as follows: 
the decision procedure returns ``valid'' for an input $J$ if and only if 
$J$ is valid. We show this. 
By the unroll collapse, the validity of $J$ is equivalent to that of $J_1,\ldots,J_k$. 
Hence $J$ is valid if and only if 
$\Search(J_i)$ returns a proof of $J_i$ for all $i$ 
by the correctness property of $\Search$ (Proposition~\ref{prop:correctness}). 
This is equivalent to that the decision procedure returns ``valid''.
\end{proof}

\begin{exa}\label{ex:decision}\rm
We show how our decision procedure works with an example 
$\Arr(1,2) * 3 \Pto (10,0) * \Ls(10,20) \vdash \Arr(1,3) * \Ls(10,20)$. 

By using the unroll collapse (Proposition \ref{prop:unroll-ls})
taking $\Ls(t,u)$ to be $\Ls(10,20)$, $\phi$ to be $\Arr(1,2) * 3 \Pto (10,0)$, and 
$\Vec{\psi}$ to be $\Arr(1,3) * \Ls(10,20)$, 
we obtain the following entailments:
\begin{align*}
(J_{a})
\quad
&
10 = 20 \land \Arr(1,2) * 3\Pto(10,0) \vdash \Arr(1,3) * \Ls(10,20),
\\
(J_{b}) 
\quad
&
\Arr(1,2) * 3\Pto(10,0) * 10\Pto(z,y) * z\Pto(20,w) \vdash \Arr(1,3) * \Ls(10,20).
\end{align*}
Then $\Search(J_{a})$ and $\Search(J_{b})$ are performed. 
$\Search(J_{a})$ immediately returns a tuple $({\rm UnsatL}, J_{a})$, 
since the antecedent is unsatisfiable. 
Computation of $\Search(J_{b})$ is done as follows:

{\bf Begin $\Search(J_{b})$}: $\Search(J_{b})$ is called.
Then $\R$ is set by ($\Pto$Ls), since it is applicable to $J_b$. 
$\Apply{\R}{(J_{b})}$ is performed. 
An instance of ($\Pto$Ls) is non-deterministically chosen:
$t\Pto (v,w)$ is taken to be $10 \Pto (z,y)$,
$\varphi$ is taken to be $\Arr(1,2) * 3\Pto(10,0) * z\Pto(20,w)$, 
$\Ls(t',u')$ is taken to be $\Ls(10,20)$,
and $\psi$ is taken to be $\Arr(1,3)$. 
Then $\Apply{\R}{(J_{b})}$ produces the following one subgoal:
\begin{align*}
  (J_{b1})\quad
  &
  \Arr(1,2) * 3\Pto(10,0) * 10\Pto(z,y) * z\Pto(20,w)
  \\
  &
  \vdash 10=20\land\Arr(1,3), \Arr(1,3)*10\Pto(z,y)*\Ls(z,20).
\end{align*}

{\bf Begin $\Search(J_{b1})$}: $\Search(J_{b1})$ is called.
Then $\R$ is set by (UnsatR). 
$\Apply{\R}{(J_{b1})}$ is performed. 
An instance of (UnsatR) is chosen, where 
$\varphi$ is taken to be the antecedent of $J_{b1}$,
$\psi$ is taken to be $10=20\land\Arr(1,3)$, and 
$\Vec{\psi}$ is taken to be $\Arr(1,3)*10\Pto(z,y)*\Ls(z,20)$. 
Then $\Apply{\R}{(J_{b1})}$ produces the following one subgoal:
\begin{align*}
  (J_{b2})\qquad
  &
  \Arr(1,2) * 3\Pto(10,0) * 10\Pto(z,y) * z\Pto(20,w)
  \\
  &
  \vdash \Arr(1,3)*10\Pto(z,y)*\Ls(z,20).
\end{align*}

{\bf Begin $\Search(J_{b2})$}: $\Search(J_{b2})$ is called.
Then $\R$ is set by ($\Pto$Ls). 
$\Apply{\R}{(J_{b2})}$ is performed. 
An instance of ($\Pto$Ls) is chosen, where 
$t\Pto (v,w)$ is taken to be $z\Pto(20,w)$,
$\varphi$ is taken to be $\Arr(1,2) * 3\Pto(10,0) * 10\Pto(z,y)$, 
$\Ls(t',u')$ is taken to be $\Ls(z,20)$,
$\psi$ is taken to be $\Arr(1,3)*10\Pto(z,y)$, and
$\Vec{\psi}$ is taken to be empty. 
Then $\Apply{\R}{(J_{b2})}$ produces the following one subgoal:
\begin{align*}
  \hspace{-1cm}(J_{b3})\quad
  &
  \Arr(1,2) * 3\Pto(10,0) * 10\Pto(z,y) * z\Pto(20,w)
  \\
  &
  \vdash z = 20 \land \Arr(1,3) * 10\Pto(z,y), z \Pto (20,w) *\Ls(20,20) * \Arr(1,3) * 10\Pto(z,y).
\end{align*}

{\bf Begin $\Search(J_{b3})$}: $\Search(J_{b3})$ is called.
Then $\R$ is set by (UnsatR). 
$\Apply{\R}{(J_{b3})}$ is performed. 
An instance of (UnsatR) is chosen, where 
$\varphi$ is taken to be the antecedent of $J_{b3}$,
$\psi$ is taken to be $z = 20 \land \Arr(1,3) * 10\Pto(z,y)$, and
$\Vec{\psi}$ is taken to be $z \Pto (20,w) *\Ls(20,20) * \Arr(1,3) * 10\Pto(z,y)$. 
Now $\varphi \land \psi$ is unsatisfiable, 
since the sizes of the required heaps by these $\varphi$ and $\psi$ are different. 

Then $\Apply{\R}{(J_{b3})}$ produces the following one subgoal:
\begin{align*}
  (J_{b4})\quad
  &
  \Arr(1,2) * 3\Pto(10,0) * 10\Pto(z,y) * z\Pto(20,w)
  \\
  &
  \vdash z \Pto (20,w) *\Ls(20,20) * \Arr(1,3) * 10\Pto(z,y).
\end{align*}

{\bf Begin $\Search(J_{b4})$}: $\Search(J_{b4})$ is called.
Then $\R$ is set by ($\Pto$LsEM). 
$\Apply{\R}{(J_{b4})}$ is performed. 
An instance of ($\Pto$LsEM) is non-deterministically chosen, where
$t\Pto(v,w)$ is taken to be $z\Pto(20,w)$,
$\varphi$ is taken to be $\Arr(1,2) * 3\Pto(10,0) * 10\Pto(z,y)$, 
$\Ls(t',u')$ is taken to be $\Ls(20,20)$, 
$\psi$ is taken to be the succedent of $J_{b4}$, 
and $\Vec{\psi}$ is taken to be empty. 
Then $\Apply{\R}{(J_{b4})}$ produces the following two subgoals:
\begin{align*}
  (J_{b51})\quad
  z \neq 20 &\land \Arr(1,2) * 3\Pto(10,0) * 10\Pto(z,y) * z\Pto(20,w)
  \\
  \vdash\quad
  &
  z \Pto (20,w) *\Ls(20,20) * \Arr(1,3) * 10\Pto(z,y),   \hbox{ and }
  \\
  (J_{b52})\quad
  z = 20 &\land \Arr(1,2) * 3\Pto(10,0) * 10\Pto(z,y) * z\Pto(20,w)
  \\
  \vdash\quad
  &
  z \Pto (20,w) *\Ls(20,20) * \Arr(1,3) * 10\Pto(z,y).  
\end{align*}

{\bf Begin $\Search(J_{b51})$}: $\Search(J_{b51})$ is called.
Then $\R$ is set by (LsElim). 
$\Apply{\R}{(J_{b51})}$ is performed. 
An instance of (LsElim) is chosen, where
$\varphi$ is taken to be the antecedent of $J_{b51}$, 
$\Ls(t',u')$ is taken to be $\Ls(20,20)$, 
$\psi$ is taken to be $z \Pto (20,w) * \Arr(1,3) * 10\Pto(z,y)$, 
and $\Vec{\psi}$ is taken to be empty. 
Then $\Apply{\R}{(J_{b51})}$ produces the following one subgoal:
\begin{align*}
  (J_{b61})\quad
  &
  z \neq 20 \land \Arr(1,2) * 3\Pto(10,0) * 10\Pto(z,y) * z\Pto(20,w)
  \\
  &
  \vdash 20 = 20 \land z \Pto (20,w) * \Arr(1,3) * 10\Pto(z,y).
\end{align*}

{\bf Begin $\Search(J_{b61})$}: $\Search(J_{b61})$ is called.
Then $\R$ is set by (Start). 
$\Apply{\R}{(J_{b61})}$ returns the empty list,
since $J_{b61}$ is evaluated to be valid by the decision procedure for {\bf SLA}
given by Theorem~\ref{thm:decidability-slar}.

{\bf End $\Search(J_{b61})$}: $\Search(J_{b61})$ returns a tuple
$({\rm Start},J_{b61})$, which is written as $\T_{b61}$. 

{\bf End $\Search(J_{b51})$}: $\Search(J_{b51})$ returns 
$({\rm LsElim},J_{b51},\T_{b61})$, which is written as $\T_{b51}$. 

{\bf Begin $\Search(J_{b52})$}: $\Search(J_{b52})$ is called.
Then $\R$ is set by ($\Pto$Ls). 
$\Apply{\R}{(J_{b52})}$ is performed. 
An instance of ($\Pto$Ls) is chosen, where
$t\Pto (v,w)$ is taken to be $z\Pto(20,w)$,
$\varphi$ is taken to be $z = 20 \land \Arr(1,2)*3\Pto(10,0)*10\Pto(z,y)$, 
$\Ls(t',u')$ is taken to be $\Ls(20,20)$, and 
$\psi$ is taken to be $z \Pto (20,w) * \Arr(1,3) * 10\Pto(z,y)$, and
$\Vec{\psi}$ is taken to be empty. 
Then $\Apply{\R}{(J_{b52})}$ produces the following one subgoal:
\begin{align*}
  (J_{b62})\quad
  z = 20 &\land \Arr(1,2) * 3\Pto(10,0) * 10\Pto(z,y) * z\Pto(20,w)
  \\
  \vdash\quad
  &
  20=20 \land z \Pto (20,w) * \Arr(1,3) * 10\Pto(z,y),
  \\
  &
  z \Pto (20,w) * 20\Pto(20,w) * \Ls(20,20) * \Arr(1,3) * 10\Pto(z,y).  
\end{align*}

{\bf Begin $\Search(J_{b62})$}: $\Search(J_{b62})$ is called.
Then $\R$ is set by (UnsatR), since the second clause in
the succedent of $J_{b62}$ is unsatisfiable under the assumption $z=20$. 
Then $\Apply{\R}{(J_{b62})}$ produces the following one subgoal:
\begin{align*}
  (J_{b72})\quad
  z = 20 &\land \Arr(1,2) * 3\Pto(10,0) * 10\Pto(z,y) * z\Pto(20,w)
  \\
  &
  \vdash
  20=20 \land z \Pto (20,w) * \Arr(1,3) * 10\Pto(z,y). 
\end{align*}

{\bf Begin $\Search(J_{b72})$}: $\Search(J_{b72})$ is called.
It immediately terminates, since
$J_{b72}$ is evaluated to be valid by the decision procedure for {\bf SLA}. 

{\bf End $\Search(J_{b72})$}: $\Search(J_{b72})$ returns a tuple
$({\rm Start},J_{b72})$, which is written as $\T_{b72}$. 

{\bf End $\Search(J_{b62})$}: $\Search(J_{b62})$ returns 
$({\rm UnsatR},J_{b62},\T_{b72})$, which is written as $\T_{b62}$. 

{\bf End $\Search(J_{b52})$}: $\Search(J_{b52})$ returns 
$(\Pto{\rm Ls},J_{b52},\T_{b62})$, which is written as $\T_{b52}$. 

{\bf End $\Search(J_{b4})$}: $\Search(J_{b4})$ returns 
$(\Pto{\rm LsEM},J_{b4},\T_{b51},\T_{b52})$, which is written as~$\T_{b4}$. 

{\bf End $\Search(J_{b3})$}: $\Search(J_{b3})$ returns 
$({\rm UnsatR},J_{b3},\T_{b4})$, which is written as $\T_{b3}$. 

{\bf End $\Search(J_{b2})$}: $\Search(J_{b2})$ returns 
$(\Pto{\rm Ls},J_{b2},\T_{b3})$, which is written as $\T_{b2}$. 

{\bf End $\Search(J_{b1})$}: $\Search(J_{b1})$ returns 
$({\rm UnsatR},J_{b1},\T_{b2})$, which is written as $\T_{b1}$. 

{\bf End $\Search(J_{b})$}: $\Search(J_{b})$ returns 
$(\Pto{\rm Ls},J_{b},\T_{b1})$, which is written as $\T_{b}$. 

Finally our decision procedure answers ``Valid'', since
each of $\Search(J_{a})$ and $\Search(J_{b})$ returns a tuple. 
\end{exa}

%% file: 9.tex
\section{Conclusion}

We have shown the decidability results for the validity checking problem
of entailments for $\SLAR$ and $\SLLA$. 
First
we have given the decision procedure for $\SLAR$ and proved its correctness under the
condition that the sizes of arrays in the succedent are not existentially quantified.
The key idea of the decision procedure is the notion of the sorted entailments. 
By using this idea, we have defined the translation $P$ of a sorted entailment into a formula in Presburger arithmetic. 
Secondly we have proved the decidability for $\SLLA$. 
The key idea of the decision procedure is to extend 
the unroll collapse technique given in~\cite{OHearn04}
to arithmetic and arrays
as well as doubly-linked list segments.
We have also given a proof system and showed correctness of the proof search algorithm for 
eliminating the list predicates in the succedent of an entailment. 

We require the condition in 
the decidability for $\SLAR$ (Theorem~\ref{thm:decidability-slar}) from a technical reason. 
It would be future work to show the decidability without this condition.

\subsection*{Acknowledgments}
This is partially supported by Core-to-Core Program (A. Advanced
Research Networks) of the Japan Society for the Promotion of Science.

%% file: bib.tex
\def\lncs#1#2#3{volume #1 of {\em Lecture Notes in Computer Science}, pages #2, #3, Springer.}
\def\berdine{Josh Berdine}
\def\ohearn{Peter W. O'Hearn}
\def\calcagno{Cristiano Calcagno}
\def\brotherston{James Brotherston}
\def\gorogiannis{Nikos Gorogiannis}
\def\kanovich{Max Kanovich}
\def\yang{Hongseok Yang}
\def\distefano{Dino Distefano}
\def\enea{Constantin Enea}
\def\saveluc{Vlad Saveluc}
\def\sighireanu{Mihaela Sighireanu}
\def\iosif{Radu Iosif}
\def\rogalewicz{Adam Rogalewicz}
\def\vojnar{Tom\'{a}\v{s} Vojnar}


%
%

